\documentclass{article}

\usepackage{amsmath}
\usepackage{amssymb}
\usepackage{amsthm}
\usepackage{microtype}
\usepackage{graphicx}
\usepackage{subfigure}
\usepackage{booktabs} % for professional tables
\usepackage{bm}
\newcommand{\algoname}[1]{\textnormal{\textsc{#1}}}

\newtheorem{theorem}{Theorem}
\newtheorem{lemma}[theorem]{Lemma}
\theoremstyle{definition}
\newtheorem{definition}{Definition}[section]
\newtheorem{claim}[theorem]{Claim}

\usepackage{hyperref}

\global\long\def\RR{\mathbb{R}}

\newcommand{\EE}{\mathbb{E}}
\usepackage[accepted]{icml2020}

\icmltitlerunning{Near Input Sparsity Time Kernel Embeddings via Adaptive Sampling}

\begin{document}
	
	\twocolumn[
	\icmltitle{Near Input Sparsity Time Kernel Embeddings via Adaptive Sampling}
	
	
	
	\icmlsetsymbol{equal}{*}
	\begin{icmlauthorlist}
		\icmlauthor{David P. Woodruff}{equal,to}
		\icmlauthor{Amir Zandieh}{equal,do}
	\end{icmlauthorlist}
	
	\icmlaffiliation{to}{Carnegie Mellon University, USA}
	\icmlaffiliation{do}{Ecole polytechnique federale de Lausanne, Switzerland. Part of this
		work was done while the author was visiting CMU}
	
	\icmlcorrespondingauthor{David P. Woodruff}{dwoodruf@cs.cmu.edu}
	\icmlcorrespondingauthor{Amir Zandieh}{amir.zandieh@epfl.ch}

	\icmlkeywords{Machine Learning, ICML}
	
	\vskip 0.3in
	]
	
	
	
	\printAffiliationsAndNotice{\icmlEqualContribution} % otherwise use the standard text.
	
	\begin{abstract}
		To accelerate kernel methods, we propose a near input sparsity time algorithm for sampling the high-dimensional feature space implicitly defined by a kernel transformation. Our main contribution is an importance sampling method for subsampling the feature space of a degree $q$ tensoring of data points in almost input sparsity time, improving the recent oblivious sketching method of (Ahle et al., 2020) by a factor of $q^{5/2}/\epsilon^2$. This leads to a subspace embedding for the polynomial kernel, as well as the Gaussian kernel, with a target dimension that is only linearly dependent on the statistical dimension of the kernel and in time which is only linearly dependent on the sparsity of the input dataset. We show how our subspace embedding bounds imply new statistical guarantees for kernel ridge regression. Furthermore, we empirically show that in large-scale regression tasks, our algorithm outperforms state-of-the-art kernel approximation methods.
	\end{abstract}
	
	\section{Introduction}
	\label{intro}
	Kernel methods provide a simple, yet powerful framework for applying non-parametric modeling techniques to a number of important problems in statistics and machine learning, such as kernel ridge regression, SVM, PCA, CCA, etc. While kernel methods are statistically well understood and perform well empirically, they often pose scalability challenges as they operate on the kernel matrix (Gram matrix) of the data, whose size scales quadratically in the size of the training dataset. Primitives such as kernel PCA or kernel ridge regression generally take a prohibitively large quadratic amount of space and at least quadratic time. Thus, much work has focused on scaling up kernel methods by producing compressed and \emph{low-rank} approximations to kernel matrices \cite{rahimi2008random, alaoui2015fast, avron2017faster, musco2017recursive, avron2017random,avron2014subspace, ahle2019oblivious, pmlr-v108-zandieh20a}.
	
	\subsection{Problem Definition}
	\label{sec:pro-def}
	For a given kernel function $k:\RR^d\times \RR^d \rightarrow \RR$ and a dataset of $d$-dimensional vectors $x_1, x_2, \cdots x_n \in \RR^d$, let $K \in \RR^{n \times n}$ be the kernel matrix corresponding to this dataset defined as $K_{ij}=k(x_i,x_j)$ for every $i,j\in [n]$. A classical solution for scaling up kernel methods is via kernel \emph{low-rank approximation}, where one seeks to find a low-rank matrix $Z\in \RR^{s\times n}$ such that $Z^\top Z$ can serve as a proxy to the kernel matrix $K$. In order to obtain statistical and algorithmic guarantees for downstream kernel-based learning applications, such as kernel regression, PCR, CCA, etc., one needs to have spectral approximation bounds on the entire surrogate kernel matrix. Formally, for given $\epsilon, \lambda> 0$ we need $Z^\top Z$ to be an \emph{$(\epsilon,\lambda)$-spectral approximation} to the kernel matrix $K$, meaning that $Z^\top Z$ has to satisfy,
	\begin{equation}\label{spectral-bound}
	\frac{K+\lambda I}{1+\epsilon} \preceq Z^\top Z+\lambda I \preceq \frac{K+\lambda I}{1-\epsilon}.
	\end{equation}
	
	Intuitively, if $\lambda$ is much larger than the operator norm of $K$ then $Z=0$ is a good solution that satisfies \eqref{spectral-bound}. On the other hand if $\lambda=0$, then the target dimension $s$ has to be at least equal to the rank of $K$. In general, the \emph{statistical dimension} (or \emph{effective dimension}) captures this tradeoff, defined as $s_\lambda :=\sum_{i=1}^n\frac{\lambda_i}{\lambda_i+\lambda}$,
	where the $\lambda_i$ are the eigenvalues of $K$.
	The goal is to find a matrix $Z \in \RR^{s\times n}$ with a target dimension $s$ which depends only linearly on $s_\lambda$, using a runtime that is nearly equal to the number of non-zero entries (i.e., the sparsity) of the input dataset, denoted by $\text{nnz}(X)$. The main motivation of this paper is the following:
	
	{\bf P :}\emph{ Given a dataset $x_1, x_2, \cdots x_n \in \RR^d$, and a kernel function $k(\cdot)$, if $K$ is the kernel matrix corresponding to this dataset with statistical dimension $s_\lambda={\bf tr}\left(K(K+\lambda I)^{-1}\right)$, can we compute a matrix $Z \in \RR^{s \times n}$ with $s = O\left( \frac{s_\lambda}{\epsilon^2} \log n\right)$, using $O\left( \text{poly}(s_\lambda,\frac{1}{\epsilon}, \log n) \cdot n + \text{poly}( \log n) \cdot \text{nnz}(X) \right)$ runtime, such that $Z^\top Z$ is an $(\epsilon,\lambda)$-spectral approximation to $K$ as per \eqref{spectral-bound}?}\\
	The runtime that {\bf (P)} is asking for requires the  $\text{poly}(s_\lambda,\epsilon^{-1})$ terms to be decoupled from the input sparsity, $\text{nnz}(X)$. Hence, up to low order terms, we aim for a runtime which only depends linearly on the sparsity of the input dataset. 
	
	We address {\bf (P)} for two important kernel classes: the degree-$q$ polynomial kernel $k(x,y)=\langle x,y\rangle^q$ for some $q\in \mathbb{Z}_+$, and the Gaussian kernel $k(x,y)=e^{-\|x-y\|_2^2/2}$. We also remark that, as we will later discuss in Section~\ref{sec:dot-prod}, our method is very general and can be applied to the class of dot-product kernels. As we will discuss in the related work section, all prior methods for approximating the polynomial kernel achieve a runtime of either the form $\text{poly}(\epsilon^{-1}, q, \log n) \cdot \text{nnz}(X)$ or $\text{poly}(s_\lambda, \epsilon^{-1}, \log n) \cdot \text{nnz}(X)$, and similarly all prior results for the Gaussian kernel achieve a runtime of either $\text{poly}(\epsilon^{-1}, r, \log n) \cdot \text{nnz}(X)$ or $\text{poly}(s_\lambda, \epsilon^{-1}, \log n) \cdot \text{nnz}(X)$, where $r$ is the radius of the input dataset. These are strictly worse than the target runtime of {\bf (P)}.
	
	\subsection{Our Results}
	We answer problem {\bf (P)} in the affirmative by designing near input sparsity time algorithms for embedding the polynomial and Gaussian kernels. Our main result for the  polynomial kernel is given in the following theorem.
	
	\begin{theorem}\label{thm:poly}
		For any dataset $x_1, \cdots x_n \in \RR^d$, any $\epsilon,\lambda>0$ and any positive integer $q$, if $K \in \RR^{n \times n}$ is the degree-$q$ polynomial kernel matrix corresponding to this dataset ($K_{i,j} := \langle x_i, x_j \rangle^q$) with statistical dimension $s_\lambda$ and $\frac{{\bf tr}(K)}{\epsilon\lambda} = O(\text{poly}(n))$, then there exists an algorithm that computes a matrix $Z \in \RR^{s \times n}$, with target dimension $s = O\left(\frac{s_\lambda}{\epsilon^2} \log n\right)$ such that, with high probability, $Z^\top Z$ is an $(\epsilon,\lambda)$-spectral approximation to $K$ as in \eqref{spectral-bound} using $O\left( {\text{poly}(\epsilon^{-1},q,\log n) \cdot s_\lambda^2n} + q^{5/2} \log^4n \cdot \text{nnz}(X)\right)$ time.
	\end{theorem}
	
	We also address ${\bf (P)}$ for approximating the Gaussian kernel by proving the following theorem.
	\begin{theorem}\label{thm:gauss}
		For any dataset $x_1, \cdots x_n \in \RR^d$ such that $\|x_i\|_2^2 \le r$ for all $i \in [n]$, any $\epsilon, \lambda \ge \frac{1}{\text{poly}(n)}$, if $K \in \RR^{n \times n}$ is the Gaussian kernel matrix corresponding to this dataset ($K_{i,j} := e^{-\|x_i-x_j\|_2^2/2}$) with statistical dimension $s_\lambda$, then there exists an algorithm that computes a matrix $Z \in \RR^{s \times n}$, with target dimension $s = O\left(\frac{s_\lambda}{\epsilon^2} \log n\right)$ such that, with high probability, $Z^\top Z$ is an $(\epsilon,\lambda)$-spectral approximation to $K$ as in \eqref{spectral-bound} using $O\left( {\text{poly}(\epsilon^{-1},r,\log n) \cdot s_\lambda^2n} + r^{5/2} \log^4n \cdot \text{nnz}(X)\right)$ time.
	\end{theorem}
	Theorems \ref{thm:poly} and \ref{thm:gauss} imply accelerated algorithms for kernel ridge regression (KRR) with improved statistical and algorithmic guarantees. We analyze the empirical risk of our sampling algorithm for the KRR problem in Appendix H.
	Furthermore, in the experiments section we evaluate our approximate KRR method on various standard large-scale regression datasets and empirically show that our method competes favorably with the  state-of-the-art, including Nystrom \cite{musco2017recursive} and Fourier features methods \cite{rahimi2008random}, as well as the oblivious sketching of \cite{ahle2019oblivious}. We show that our method achieves better testing error and smaller runtime on large  datasets with more than half a million training examples.
	
	{\bf Additional downstream learning applications:}
	While we focus on KRR here, we remark that spectral approximation bounds form the basis of analyzing sketching methods for tasks including kernel low-rank approximation, PCA, CCA, k-means and many more. In the kernelized setting, such bounds have been analyzed, without regularization, for the polynomial kernel~\cite{avron2014subspace}. It is shown in \cite{cohen2017input} that \eqref{spectral-bound} along with a trace condition on $Z^\top Z$ (which holds for the sampling approaches we consider) yields a so-called \emph{projection-cost preservation} condition. With $\lambda$ chosen appropriately, this condition ensures that $Z^\top Z$ can serve as a proxy for $K$ for approximately solving kernel k-means and for certain versions of kernel PCA and kernel CCA. See~\cite{musco2017recursive} for details, where this analysis is carried out for the Nystrom method.
	
	\subsection{Prior Work}
	A popular approach for accelerating kernel methods is based on Nystrom sampling. We refer the reader to the work of \cite{musco2017recursive} and the references therein. By recursively sampling Nystrom landmarks according to the so-called ridge leverage score distribution, \citet{musco2017recursive} prove that for any kernel $K$ with statistical dimension $s_\lambda$, there exists an algorithm that outputs a matrix $Z \in \RR^{s \times n}$ with $s=O\left(\frac{s_\lambda}{\epsilon} \log n\right)$ which satisfies the spectral approximation guarantee of \eqref{spectral-bound} with high probability, using $O\left( n  \frac{s_\lambda^2}{\epsilon^2} \cdot \log^2 n + \frac{s_\lambda}{\epsilon} \log n \cdot \text{nnz}(X) \right)$ runtime. However, the leading term in the time complexity of this method is $O \left(\frac{s_\lambda}{\epsilon} \log n \cdot \text{nnz}(X)\right)$, which unsatisfactorily depends on $\epsilon^{-1}$ and also depends linearly on $s_\lambda$.
	Hence, for both the polynomial and Gaussian kernels our Theorems~\ref{thm:poly}~and~\ref{thm:gauss} improve on the runtime of this method by a factor of $\epsilon^{-1} s_\lambda$.
	
	Another popular line of work on kernel approximation problems is the Fourier features method of \citet{rahimi2008random}. It is proved in \cite{avron2017random} that this method can achieve spectral approximation guarantees for the Gaussian kernel using a sub-optimal number $s \approx \epsilon^{-2} \frac{n}{\lambda} \log n$ of samples and $O\left( \epsilon^{-2} \frac{n}{\lambda} \log n \cdot \text{nnz}(X) \right)$ runtime. This sample complexity is substantially larger than our result in Theorem~\ref{thm:gauss}. Furthermore we improve the runtime of this method by a factor of $\epsilon^{-2} \frac{n}{\lambda}$. However, \cite{avron2017random} show that this method can be modified to achieve a sample complexity of $s=\Theta(1)^d \cdot \frac{s_\lambda}{\epsilon^2} \log n$ using a runtime of $\Theta(1)^d \cdot \frac{s_\lambda}{\epsilon^2} \log n \cdot \text{nnz}(X)$. For constant dimensional datasets (constant $d$) the number of samples that \cite{avron2017random} achieve is comparable to our target dimension in Theorem~\ref{thm:gauss} but it deteriorates exponentially with the dimension $d$. Furthermore, the runtime of this method is substantially larger than our runtime by a factor of $\Theta(1)^d \cdot \epsilon^{-2} s_\lambda$.
	
	In the linear sketching literature, \cite{avron2014subspace} proposed an \emph{oblivious subspace embedding} for the polynomial kernel based on the TensorSketch of \cite{pham2013fast}. They applied this method to a wide array of kernel problems, including PCA, PCR, and CCA. The runtime of this method, while nearly linear in $\text{nnz}(X)$, scales \emph{exponentially} in the degree $q$ of the polynomial kernel. Their runtime for the degree-$q$ polynomial kernel is $O\left( \frac{q \cdot 3^q s_\lambda^2}{\epsilon^2} + q \cdot  \text{nnz}(X)\right)$, which has an unsatisfactory $3^q$ term.
	
	Recently, \cite{ahle2019oblivious} proposed a new \emph{oblivious sketching} solution for the polynomial kernel that improves the exponential dependence of TensorSketch on $q$ to polynomial.
	\citet{ahle2019oblivious} gave an algorithm that outputs a matrix $Z \in \RR^{s \times n}$ with $s=\widetilde{O}\left(\frac{q^4s_\lambda}{\epsilon^2} \right)$ which satisfies the spectral approximation guarantee of \eqref{spectral-bound} with high probability. Their algorithm has  $\widetilde{O}\left( \frac{q^5 s_\lambda}{\epsilon^2} \cdot n + \frac{q^5}{\epsilon^2}\cdot \text{nnz}(X) \right)$ runtime\footnote{$\widetilde{O}$ notation hides $\text{poly}(\log n)$ factors.}. This runtime has an undesirable inverse polynomial dependence on $\epsilon$ and scales sub-optimally with the degree of the polynomial kernel as $q^5$. Our Theorem~\ref{thm:poly} improves the runtime of \cite{ahle2019oblivious} by an $\epsilon^{-2}q^{5/2}$ factor.
	Moreover, they showed that their sketch for the polynomial kernel leads to an efficient oblivious sketch for the Gaussian kernel on bounded datasets.
	\citet{ahle2019oblivious} gave an algorithm that for any dataset $x_1,x_2,\cdots x_n \in \RR^d $ with radius $r$, computes a matrix $Z\in\RR^{s \times n}$ with $s=\widetilde{O}\left(\frac{r^5s_\lambda}{\epsilon^2}\right)$ which spectrally approximates the Gaussian kernel matrix corresponding to this dataset as in \eqref{spectral-bound} with high probability. This was the first result that resolved the curse of dimensionality for embedding the high dimensional Gaussian kernel. The algorithm has $\widetilde{O}\left( \frac{r^6 s_\lambda}{\epsilon^2} \cdot n+ \frac{r^6}{\epsilon^2}\cdot \text{nnz}(X) \right)$ runtime, which unsatisfactorily depends on $1/\epsilon^2$ and scales poorly as a function of the dataset's radius as $r^6$. Our Theorem~\ref{thm:gauss} improves this runtime by a factor of $\epsilon^{-2}r^{7/2}$.
	
	\subsection{Our Techniques}
	Our method relies on the fact that any kernel function $k:\RR^d \times \RR^d \rightarrow \RR$ defines a lifting $\phi$ such that the kernel function computes the inner product between the lifted data points, i.e., $k(x,y) = \langle \phi(x), \phi(y) \rangle$. Therefore, any kernel matrix $K$ can be decomposed as $K = \Phi^\top \Phi$ where $\Phi$ is a matrix with $n$ columns whose columns are the lifted data points $\phi(x_i)$. Our approach is to design an importance sampling matrix $\Pi$ such that $Z = \Pi \Phi$ satisfies the spectral approximation guarantee of \eqref{spectral-bound}. Our algorithm generates a sampling matrix $\Pi$ that samples a small number of rows of $\Phi$ using a \emph{recursive leverage score} sampling technique, which has been extensively applied to various algorithmic problems in the literature \cite{kapralov2014single,alaoui2015fast, cohen2016online, musco2017recursive, avron2017random,cohen2017input}.
	Our main novelty is in generating a sample from the leverage score distribution without ever forming the entire distribution explicitly, as the support size of this distribution is equal to the number of rows of $\Phi$ which is typically high (even infinite).
	
	For the polynomial kernel, the lifting matrix is $\Phi = X^{\otimes q}$, where $X^{\otimes q}$ is a $d^q \times n$ matrix whose columns are obtained by a $q$-fold self-tensoring of the columns of the dataset matrix $X \in \RR^{d \times n}$ (see Section~\ref{sec:prelim} for notation). After multiple reductions, our importance sampling problem boils down to performing $\ell_2$-sampling on a vector of the form $X^{\otimes q} v$, where $v$ is an arbitrary vector in $\RR^{n}$. Here by $\ell_2$-sampling of a vector, we mean sampling a coordinate proportional to its squared value. We design a primitive that can generate a sample $i \in [d^q]$ with probability proportional to the squared value of the $i^{th}$ entry of the vector $X^{\otimes q} v$ using roughly  $\text{nnz}(X)$ time. Our algorithm relies on the fact that, by reshaping, entries of the vector $X^{\otimes q} v$ are in bijective correspondence with entries of the matrix $X^{\otimes q-1} \cdot \text{diag}(v) X^\top$, where $\text{diag}(v)$ is a diagonal $n\times n$ matrix whose diagonal entries are the elements of $v$. Therefore, our importance sampling amounts to sampling an element of $X^{\otimes q-1} \cdot \text{diag}(v) X^\top$ with probability proportional to the square of its absolute value. We do this by first sampling a column of this matrix with probability proportional to its squared norm, and then sampling a row with probability proportional to the squares of the entries of the sampled column.
	After sampling a column $l \in [d]$ of the matrix $X^{\otimes q-1} \cdot \text{diag}(v) X^\top$, we next perform $\ell_2$-sampling on the $l^{th}$ column of the mentioned matrix, which is in the form of $X^{\otimes q-1} u$, where $u=\text{diag}(v) X_{l,\star}^\top$. One can see that we have made progress and now it is enough to iterate in this fashion by performing $\ell_2$-sampling on $X^{\otimes q-1} u$.
	However, note that $X^{\otimes q-1} \cdot \text{diag}(v) X^\top$ has $d^{q-1}$ rows, and hence, computing its column norms is prohibitively expensive. We tackle this issue by sketching the columns of $X^{\otimes q-1} \cdot \text{diag}(v) X^\top$ using the sketch introduced in \cite{ahle2019oblivious}, which is able to preserve the column norms up to a small error and with runtime roughly $\text{nnz}(X)$.
	
	Our algorithm is actually more involved and includes extra dimensionality reduction steps. In the paragraph above we explained how to generate a single sample with the right distribution, but in order to obtain the spectral approximation guarantee of \eqref{spectral-bound} we need to generate $s=O\left(\frac{s_\lambda}{\epsilon^2}\log n\right)$ such samples. It is crucial that our runtime does not lose a multiplicative factor of $s$. We heavily exploit the structure of tensor products to reuse most computations and generate $s$ samples in time proportional to $\text{nnz}(X)$.
	Moreover, to spectrally approximate the Gaussian kernel, we adapt our sampling algorithm to a truncated Taylor expansion of the Gaussian kernel. Furthermore, in Section~\ref{sec:dot-prod} we discuss how our method can be generalized to any dot-product kernel.
	
	\section{Preliminaries}\label{sec:prelim}
	Throughout the paper, for any matrices $A \in \RR^{m \times n}$ and $B\in \RR^{d \times n}$, $A\oplus B \in \RR^{(m+d) \times n}$ denotes the vertical concatenation of $A$ and $B$, i.e.,
	$A \oplus B = \begin{bmatrix}
	A\\
	B
	\end{bmatrix}$.\\
	Moreover, $A \otimes B \in \RR^{(md) \times n}$ denotes the vertical tensor product of $A$ and $B$. The rows of $A\otimes B$ are indexed by $(i,j)$ where $i \in [m]$ and $j \in [d]$ and for any $l \in [n]$, $[A\otimes B]_{(i,j),l}= A_{i,l} \cdot B_{j,l}$.
	We also use $A^{\otimes q}$ to denote, $A^{\otimes q} = \underbrace{A\otimes A \cdots \otimes A}_{\text{$q$ terms}}$.\\
	For any matrix $X$ we use $X_{i,\star}$ to denote its $i^{th}$ row and we use $X_{\star,i}$ to refer to its $i^{th}$ column. Also for any set $S$, $X_{S,\star}$ denotes a sub-matrix of $X$ that includes rows $i \in S$ of $X$.
	
	\section{Algorithm and Analysis}
	Let $\Phi \in \RR^{D \times n}$ be the feature matrix whose columns are the projections of the data points in the feature space. We start by presenting a recursive importance sampling algorithm that efficiently computes a matrix $Z$ which satisfies the spectral approximation guarantee of \eqref{spectral-bound} for the kernel $K=\Phi^\top \Phi$.
	Sampling rows of $\Phi$ with probabilities proportional to the squared row norms of the matrix $\Phi (\Phi^\top \Phi + \lambda I)^{-1/2}$, which are known as the \emph{ridge leverage scores} of $\Phi$, is an efficient sampling strategy for obtaining the spectral approximation guarantee of \eqref{spectral-bound}. In Algorithm~\ref{alg:outerloop}, we give a generic recursive method for performing approximate leverage score sampling on any matrix $\Phi$. The recursive procedure works by generating samples from a crude approximation to the leverage scores and iteratively refining the sampling distribution.
	\begin{algorithm}[t]
		\caption{\algoname{Recursive Leverage Score Sampling}}
		{\bf input}: Matrix $\Phi \in \RR^{D \times n}$, $\lambda\in \RR_{+}$, $\epsilon \in \RR_{+}$, $\mu\in \RR_{+}$\\
		{\bf output}: Sampling matrix $\Pi \in \RR^{s \times D}$
		\begin{algorithmic}[1]
			\STATE{$s \gets C \frac{\mu}{\epsilon^2} \log_2 n$ for some constant $C$}
			\STATE{$S_0 \gets \{0\}^{1\times D}$}
			\STATE{$\lambda_0 \gets \frac{\|\Phi\|_F^2}{\epsilon}$}
			\STATE{$T\gets \lceil \log_2 \frac{\lambda_0}{\lambda} \rceil$}
			\FOR{$t = 1$ to $T$}
			\STATE{$S_{t} \gets \textsc{RowSampler}\left(\Phi, S_{t-1}\Phi, \lambda_{t-1},s\right)$} \label{oversample-alg}
			\STATE{$\lambda_{t} = \lambda_{t-1}/2$}
			\ENDFOR
			\STATE{\textbf{return} $\Pi = S_T$}
		\end{algorithmic}
		\label{alg:outerloop}
	\end{algorithm}
	We first introduce the definition of a \emph{row norm sampler} as follows,
	\begin{definition}[Row Norm Sampler]\label{def:row-samp}
		Let $\Phi$ be a $D\times n$ matrix with rows $\phi_1,\phi_2, \cdots \phi_D \in \RR^n$. For any probability distribution $\{p_i\}_{i=1}^D$ that satisfies $p_i\ge \frac{1}{4}\frac{\|\phi_i\|_2^2}{\|\Phi\|_F^2}$ for all $i\in[D]$, and any positive integer $s$, a \emph{rank-$s$ row norm sampler} for matrix $\Phi$ is a random matrix $S\in\RR^{s \times D}$ which is constructed by generating $s$ i.i.d. samples $j_1, j_2, \cdots j_s \in [D]$ with distribution $\{p_i\}_{i=1}^D$ and letting the $r^{th}$ row of $S$ be $\frac{1}{\sqrt{sp_{j_r}}}{\bf e}_{j_r}^\top$ for every $r\in [s]$, where ${\bf e}_{1}, {\bf e}_{2}, \cdots {\bf e}_{D}\in \RR^D$ are the standard basis vectors in $\RR^D$.
	\end{definition}
	Now we are ready to prove the correctness of Algorithm \ref{alg:outerloop},
	\begin{lemma}\label{resursive-rlss-lem}
		Suppose that for any matrices $\Phi\in \RR^{D \times n}$ and $B\in \RR^{m\times n}$, any  $\lambda'>0$, and any positive integer $s'$, the primitive \textsc{RowSampler}$(\Phi,B,\lambda',s')$ returns a rank-$s'$ row norm sampler for matrix $\Phi(B^\top B + \lambda' I)^{-1/2}$ as in Definition \ref{def:row-samp}. Then for any matrix $\Phi \in \RR^{D \times n}$ with statistical dimension $s_\lambda = \|\Phi(\Phi^\top \Phi + \lambda I)^{-1/2}\|_F^2$, any $\lambda, \epsilon>0$, any $\mu \ge s_\lambda$, Algorithm \ref{alg:outerloop} returns a sampling matrix $\Pi \in \RR^{s \times d}$ with $s = O(\frac{\mu}{\epsilon^2} \log n)$ such that with probability $1 - \frac{1}{\text{poly}(n)}$,
		\[ \frac{\Phi^\top \Phi + \lambda I}{1+\epsilon} \preceq \Phi^\top \Pi^\top \Pi \Phi + \lambda I \preceq \frac{\Phi^\top \Phi + \lambda I}{1-\epsilon}.\]
	\end{lemma}
	The proof of this lemma is included in Appendix C.
	
	\subsection{Adaptive Sampling for the Polynomial Kernel}\label{sec:polynomial}
	The polynomial kernel of degree $q$ is defined as $k(x,y)=\langle x,y \rangle^q$. Using the definition of tensor products, one can see that $\langle x,y \rangle^q = \langle x^{\otimes q},y^{\otimes q} \rangle$, where $x^{\otimes q}$ and $y^{\otimes q}$ are $q$-fold self tensor products of vectors $x$ and $y$, respectively.
	Suppose $X \in \RR^{d \times n}$ is the dataset matrix. The polynomial kernel matrix can be decomposed as $K=\left(X^{\otimes q}\right)^\top X^{\otimes q}$, where $X^{\otimes q}$ is a $d^q \times n$ matrix whose columns are obtained by the $q$-fold self tensoring of the columns of $X$. The goal is to apply the iterative leverage score sampling of Algorithm \ref{alg:outerloop} to the feature matrix $\Phi=X^{\otimes q}$ in nearly $\text{nnz}(X)$ time. Note that the matrix $\Phi$ has a large number $d^q$ of rows so even assuming that an oracle gives us the leverage score distribution of $\Phi$ for free, just reading this distribution takes $d^q$ time. We show how to generate samples from the right distribution quickly.\\
	Algorithm \ref{alg:outerloop} crucially uses the primitive \textsc{RowSampler}, which carries out the main computations of our proposed algorithm. This primitive performs row norm sampling (see Definition \ref{def:row-samp}) on a matrix of the form $\Phi(B^\top B + \lambda I)^{-1/2}$, for any matrix $B$, very efficiently.
	
	\subsubsection{\textsc{RowSampler} for the Polynomial kernel}
	An important technical contribution of this work is an efficient algorithm that can perform row norm sampling (see Definition \ref{def:row-samp}) on a matrix of the form $X^{\otimes q} (B^\top B + \lambda I)^{-1/2}$ using nearly $\text{nnz}(X)$ runtime, where $X\in \RR^{d\times n}$ and $B\in \RR^{m\times n}$. Our primitive uses the sketch which was proposed in \cite{ahle2019oblivious} to preserve the norm of vectors in $\RR^{d^q}$ and sketch vectors of the form $x^{\otimes q}$ quickly. The next lemma follows from Theorem 1.2 of \cite{ahle2019oblivious},
	\begin{lemma}\label{soda-result}
		For every positive integers $q,d$, every $\epsilon>0$, and every $\delta>0$, there exists a distribution on random matrices $Q^q \in \RR^{m \times d^q}$ with $m = O\left(\frac{1}{\epsilon^2}\log \frac{1}{\delta} \right)$ such that,
		$\Pr\left[ \|Q^q y\|_2^2 \in (1\pm \epsilon)\|y\|_2^2 \right] \ge 1 - \delta$ for any $y \in \RR^{d^q}$.
		Moreover, for any $x \in \RR^d$, the total time to compute $Q^q \left(x^{\otimes q-j} \otimes {\bf e}_1^{\otimes j}\right)$ for all $j=0,1,2, \cdots q$ is $O\left( {\frac{q^2}{\epsilon^4}} \log^4\frac{1}{\delta} + \frac{q^{3/2}}{\epsilon} \log \frac{1}{\delta} \cdot \text{nnz}(x) \right)$, where ${\bf e}_1\in\RR^d$ is the standard basis vector along the first coordinate.
	\end{lemma}
	We prove this lemma in Appendix D.
	Now we are ready to design the procedure \textsc{RowSampler} to perform row norm sampling on matrices of the form $X^{\otimes q} (B^\top B+\lambda I)^{-1/2}$.

	\begin{algorithm}[h!]
		\caption{\algoname{RowSampler for Polynomial Kernel}}
		{\bf input}: $X \in \RR^{d \times n}$, $q\in \mathbb{Z}_+$, $B \in \RR^{m\times n}$, $\lambda \in \RR_{+}$, $s\in \mathbb{Z}_{+}$\\
		{\bf output}: Sampling matrix $S \in \RR^{s \times d^q}$
		\begin{algorithmic}[1]
			\STATE{Generate $H\in\RR^{n\times d'}$ with i.i.d. normal entries with $d'= C_1 q \log_2n$}
			\STATE{$M \gets (B^\top B + \lambda I)^{-1/2} \cdot H$}\label{M-alg}
			\STATE{Let $Q^q\in\RR^{m'\times d^q}$ be an instance of the sketch from Lemma \ref{soda-result} with $\epsilon=\frac{1}{10q}$, $\delta=\frac{1}{\text{poly}(n)}$, $m'=C_2q^2 \log_2n$}
			\STATE{Compute $P_j=Q^q \left(X^{\otimes (q-j)} \otimes E_1^{\otimes j}\right)$ for all $j=0,1, \cdots q-1$, where $E_1 \in \RR^{d\times n}$ is a matrix whose columns are copies of ${\bf e}_1$, i.e., $E_1=\Big[ \underbrace{ {\bf e}_1, {\bf e}_1, \cdots {\bf e}_1 }_{n \text{ copies}}  \Big]$}\label{Pj-alg}
			\STATE{$Z\gets P_0M$}
			\STATE{$p_i\gets \frac{\|Z_{\star,i}\|_2^2}{\|Z\|_F^2}$ for every $i\in[d']$}\label{J-dist-alg}
			
			\STATE{Generate i.i.d. samples $j_1,j_2,\cdots j_s\in[d']$ with distribution $\{p_i\}_{i=1}^{d'}$} \label{jsamples-alg}
			
			\STATE{Let $h:[d]\rightarrow[s']$ be a fully independent and uniform hash function with $s'=\lceil {q}^{3/2}s \rceil$}
			\STATE{Let $h^{-1}(r)=\left\{ j\in[d]:h(j)=r \right\}$ for every $r\in[s']$}
			
			\STATE{For every $r\in[s']$, generate $G_r \in \RR^{n'\times d_r}$ with i.i.d. normal entries where $d_r=|h^{-1}(r)|$ and $n'=C_3q^2\log_2n$}
			\STATE{$W_r \gets G_r \cdot X_{h^{-1}(r),\star}$ for every $r\in[s']$}\label{W}
			\FOR{$l=1$ to $s$}
			\STATE{$D^{(0)} \gets \text{diag}(M_{\star,j_l})$}
			\FOR{$a=1$ to $q$}
			\STATE{$p^{a}_r\gets \frac{\left\|W_r \cdot D^{(a-1)} \cdot P_a^\top \right\|_F^2}{\sum_{t=1}^{s'} \left\|W_t \cdot D^{(a-1)} \cdot P_a^\top \right\|_F^2}$ for every $r \in [s']$}\label{dist-par}
			\STATE{Generate a sample $t$ with distribution $\{p^a_r\}_{r=1}^{s'}$}
			\STATE{$q^a_i \gets \frac{\left\|X_{i,\star}  D^{(a-1)}  P_a^\top \right\|_2^2}{\left\|X_{h^{-1}(t),\star}  D^{(a-1)}  P_a^\top \right\|_F^2}$ for all $i\in h^{-1}(t)$}\label{dist-qai}
			
			\STATE{Sample an $i_a$ with distribution $\{q^a_i \}_{i\in h^{-1}(t)}$}\label{isample-alg}
			
			\STATE{$D^{(a)}\gets D^{(a-1)}\cdot \text{diag}(X_{i_a,\star})$}
			\ENDFOR
			
			\STATE{$\beta \gets 0$}
			\FOR{$j=1$ to $d'$}\label{beta-forloop}
			\STATE{$L^{(0)} \gets \text{diag}(M_{\star,j})$}
			\FOR{$b=1$ to $q$}
			\STATE{$p^*_b \gets \frac{\left\|W_{h(i_b)} \cdot L^{(b-1)} \cdot P_b^\top \right\|_F^2}{\sum_{t=1}^{s'} \left\|W_t \cdot L^{(b-1)} \cdot P_b^\top \right\|_F^2}$}\label{p-star}
			\STATE{$q^*_b \gets \frac{\left\|X_{i_b,\star}  L^{(b-1)}  P_b^\top \right\|_2^2}{\left\|X_{h^{-1}(h(i_b)),\star}  L^{(b-1)}  P_b^\top \right\|_F^2}$}\label{q-star}
			\STATE{$L^{(b)}\gets L^{(b-1)}\cdot \text{diag}(X_{i_b,\star})$}
			\ENDFOR
			\STATE{$\beta \gets \beta + sp_j \cdot \prod_{b=1}^q (p^*_b q^*_b) $}
			\ENDFOR
			
			\STATE{Let $l^{th}$ row of $S$ be $ \beta^{-1/2}  
				\left({\bf e}_{i_1} \otimes {\bf e}_{i_2} \otimes \cdots {\bf e}_{i_q}\right)^\top$}
			\ENDFOR
			\STATE{\textbf{return} $S$}
		\end{algorithmic}
		\label{alg:rotatedrowsampler-poly}
	\end{algorithm}
	
	{\bf Overview of Algorithm \ref{alg:rotatedrowsampler-poly}:} The goal is to generate a sample $(i_1,i_2, \cdots i_q) \in [d]^q$ with probability proportional to the squared norm of the row $(i_1, \cdots i_q)$ of the matrix $X^{\otimes q} (B^\top B + \lambda I)^{-1/2}$. Because the matrix $(B^\top B + \lambda I)^{-1/2}$ is of a large $n \times n$ size, we seek to compress it without perturbing the row norm distribution of $X^{\otimes q} (B^\top B + \lambda I)^{-1/2}$. This can be done by applying a JL-transformation to the rows of this matrix (see, e.g., \cite{dasgupta2003elementary,kane2014sparser}). Let $H \in \RR^{n \times d'}$ be a random matrix with i.i.d. normal entries with $d' = C_1 q \log_2 n$. Then with probability $1 - \frac{1}{\text{poly}(n^q)}$ the norm of each row of the matrix $X^{\otimes q} (B^\top B + \lambda I)^{-1/2} \cdot H$ will be preserved up to a $(1\pm 0.1)$ factor and hence by a union bound, with high probability all row norms of $X^{\otimes q} (B^\top B + \lambda I)^{-1/2} \cdot H$ are within a $(1\pm 0.1)$ factor of the row norms of the original matrix. This is done in line~2 of the algorithm by computing the matrix $M = (B^\top B + \lambda I)^{-1/2} \cdot H$, which can be done quickly since $B$ is a low rank matrix and $H$ has few columns. 
	
	Now the problem is reduced to performing row norm sampling on $X^{\otimes q} M$. In order to generate a sample with distribution proportional to the squares of the row norms of $X^{\otimes q} M$ we can first sample a column of this matrix with probability proportional to the squared column norms and then generate a row index with probability proportional to the squared values of the entries of the selected column. This process generates a random index with our desired distribution. Computing the exact column norms of $X^{\otimes q} M$ is too expensive as this matrix has $d^q$ rows, but if we apply the sketch $Q^q$ from Lemma \ref{soda-result}, we can compress the rows while preserving the column norms, in near input sparsity time, up to small error. So, it is enough to sample a column $j$ with probability proportional to the squared column norms of $Q^q X^{\otimes q} M$, which is done in lines~3-7 of the algorithm. 
	
	Given that the $j^{th}$ column of $X^{\otimes q} M$ was sampled, all we need to do is sample an entry of $X^{\otimes q} M_{\star,j}$ with probability proportional to the squared values of its entries. Note that forming this vector is out of the question since it has $d^q$ coordinates. By basic properties of tensor products, the entries of $X^{\otimes q} M_{\star,j}$ are in bijective correspondence with the entries of the matrix $X \cdot \text{diag}(M_{\star,j}) \cdot \left(X^{\otimes q-1}\right)^\top$, where entry $(i_1,i_2, \cdots i_q)$ of $X^{\otimes q} M_{\star,j}$ is equal to the entry at row $i_1$ and column $(i_2, \cdots i_q)$ of $X \cdot \text{diag}(M_{\star,j}) \cdot \left(X^{\otimes q-1}\right)^\top$. Therefore, it is enough to sample an entry of the matrix $X \cdot \text{diag}(M_{\star,j}) \cdot \left(X^{\otimes q-1}\right)^\top$ with probability proportional to its squared value. To this end, we first sample a row of this matrix with probability proportional to the squared row norms, and then sample a column by performing $\ell_2$-sampling on the sampled row. Since $X \cdot \text{diag}(M_{\star,j}) \cdot \left(X^{\otimes q-1}\right)^\top$ has a large number $d^{q-1}$ of columns, we first sketch the rows of this matrix, incurring only a factor $\left(1\pm \frac{1}{10q}\right)$ perturbation to the row norms, and then perform row norm sampling on the sketched matrix. Now we have an index $i_1 \in [d]$ sampled from the right distribution and all that is left to do is to carry out $\ell_2$-sampling on the vector $X_{i_1, \star} \cdot \text{diag}(M_{\star,j}) \cdot \left(X^{\otimes q-1}\right)^\top$. Note that we have made progress because this vector has size $d^{q-1}$ and we have reduced the size by a factor of $d$. We recursively repeat this process of reshaping the tensor product to a matrix and sampling a row of the matrix $q$ times until having all $q$ indices $i_1,i_2, \cdots i_q$. Algorithm~\ref{alg:rotatedrowsampler-poly} does this. Note that the actual procedure requires more work because we need to generate $s$ i.i.d. samples with the row norm distribution. To ensure that our runtime does not lose a multiplicative factor of $s$, resulting in $s\cdot \text{nnz}(X)$ total time, we need to do extra sketching and a random partitioning of the rows of the matrix $X$ to $\Theta({q}^{3/2} s)$ buckets.
	The formal guarantee on Algorithm~\ref{alg:rotatedrowsampler-poly} is given in the following lemma.
	
	\begin{lemma}\label{lem:rotatedrowsampler-poly}
		For any matrices $X\in \RR^{d\times n}$ and $B\in\RR^{m\times n}$, any $\lambda>0$ and any positive integers $q,s$, with high probability, Algorithm \ref{alg:rotatedrowsampler-poly} outputs a ranks-$s$ row norm sampler for $X^{\otimes q}(B^\top B + \lambda I)^{-1/2}$ (Definition \ref{def:row-samp}) in time $O\left( m^2 n + q^{15/2} s^2 n \log^3 n + q^{5/2} \log^3n \cdot \text{nnz}(X) \right)$.
	\end{lemma}
	\begin{proof}
		All rows of the sampling matrix $S\in\RR^{s\times d^q}$ (output of Algorithm \ref{alg:rotatedrowsampler-poly}) have independent and identical distributions because the algorithm generates i.i.d. samples  $j_1,j_2, \cdots j_s$ in line~7 and then for each  $l\in[s]$, the $l^{th}$ row of the matrix $S$ is constructed by sampling $i_1,i_2,\cdots i_q$ in line~18 from a distribution that is solely determined by $j_l$ and is independent of the values of $j_{l'}$ for $l'\neq l$.
		
		Since every row of $S$ is identically distributed, let us consider the distribution of the $l^{th}$ row of $S$ for some arbitrary $l\in[s]$. Let $J$ be a random variable that takes values in $\{1,2, \cdots d'\}$ with probability distribution $\{p_i\}_{i=1}^{d'}$ defined in line~6 of Algorithm \ref{alg:rotatedrowsampler-poly}. The random index $j_l$ generated in line~7 of the algorithm is a copy of the random variable $J$.
		For any $j\in[d']$, let $I^j=(I^j_1,I^j_2, \cdots I^j_q)$ be a vector-valued random variable that takes values in $[d]^q$ with the following conditional probability distribution for every $a=1,2, \cdots q$,
		\begin{align*}
		&\Pr[I^j_a =i | I^j_1=i_1, I^j_2=i_2, \cdots I^j_{a-1}=i_{a-1}]\\
		& = \frac{\|W_{h(i)} \cdot D^{j,a-1} P_a^\top \|_F^2}{\sum_{t=1}^{s'} \|W_{t}  D^{j,a-1} P_a^\top \|_F^2}  \frac{\| X_{i,\star} D^{j,a-1}P_a^\top \|_2^2}{\| X_{h^{-1}(h(i)),\star} D^{j,a-1}P_a^\top \|_F^2},
		\end{align*}
		where $W_r$ for every $r\in[s']$ are the matrices defined in line~11 of the algorithm and $D^{j,a-1}$ is a diagonal matrix of size $n\times n$ whose diagonal entries are $D^{j,a-1}_{rr} = M_{r,j} \cdot \prod_{b=1}^{a-1} X_{i_b,r}$,
		for every $r \in [n]$ and $a\in[q]$. For ease of notation we drop the superscript $j$ and just write $D^{a-1}$.
		One can verify that the vector random variable $(i_1,i_2, \cdots i_q)$ obtained by stitching together the random indices generated in line~18 of the algorithm, is a copy of the random variable $I^{j_l}$.
		
		Let $\beta$ be the quantity that the for loop in lines~21-30 of the algorithm computes. If $i_1,i_2, \cdots i_q \in[d]$ are the indices sampled in line~18 of the algorithm, then the value of $\beta$ can be computed as, $\beta = s\sum_{j=1}^{d'} p_j \prod_{b=1}^q p^*_b q^*_b$,
		where $p^*_b = \frac{\|W_{h(i_b)} \cdot D^{b-1} P_b^\top \|_F^2}{\sum_{t=1}^{s'} \|W_{t}  D^{b-1} P_b^\top \|_F^2}$ and $q^*_b = \frac{\| X_{i_b,\star} D^{b-1}P_b^\top \|_2^2}{\| X_{h^{-1}(h(i_b)),\star} D^{b-1}P_b^\top \|_F^2}$ are the quantities computed in lines~25 and 26 of the algorithm.
		Hence, for any $i_1,i_2, \cdots i_q \in [d]$, the distribution of $S_{l,\star}$ is,
		\begin{align}
		&\Pr\left[ S_{l,\star}= \beta^{-1/2} ({\bf e}_{i_1} \otimes {\bf e}_{i_2} \otimes \cdots {\bf e}_{i_q})^\top \right]\nonumber\\
		&= \sum_{j=1}^{d'} \Pr\left[ \left. S_{l,\star}= \beta^{-1/2} ({\bf e}_{i_1} \otimes \cdots {\bf e}_{i_q})^\top \right| J=j \right] \cdot p_j \nonumber\\
		&= \sum_{j=1}^{d'} \Pr\left[ I^j=(i_1 , i_2, \cdots i_q) \right] \cdot p_j. \label{dist-sampling-matrix-poly}
		\end{align}
		By the law of total probability, we have $\Pr\left[ I^j=(i_1 , i_2, \cdots i_q) \right] = \prod_{a=1}^{q} \Pr [I^j_a =i_a | I^j_1=i_1, \cdots I^j_{a-1}=i_{a-1} ]$, and therefore, because $\Pr \left[I^j_a =i_a | I^j_1=i_1, \cdots I^j_{a-1}=i_{a-1} \right] = p^*_a q^*_a$, we find that 
		\[\Pr\left[ S_{l,\star}= \beta^{-1/2} ({\bf e}_{i_1} \otimes {\bf e}_{i_2} \otimes \cdots {\bf e}_{i_q})^\top \right] = \frac{\beta}{s}.\] 
		Now note that for any $r\in[s']$, $W_r$ is defined as $W_r=G_r \cdot X_{h^{-1}(r),\star}$ where $G_r$ is a matrix with i.i.d. normal entries with $n'=C_3q^2\log_2n$ rows. Therefore, $G_r$ is a JL-transform and for every $a\in[q], r\in[s']$, with high probability, i.e.,
		\begin{equation}\label{jl-poly}
		\frac{\|W_{r} D^{a-1} P_a^\top \|_F^2}{n'} \in \frac{\| X_{h^{-1}(r),\star} D^{a-1} P_a^\top \|_F^2}{1 \pm {1}/{10q}}.
		\end{equation}
		For a simple proof of \eqref{jl-poly}, see \cite{dasgupta2003elementary} (see also \cite{kane2014sparser} for a more efficient version). By union bounding over $qs'd'$ events, \eqref{jl-poly} holds simultaneously for all $a\in[q]$, $j\in[d']$, and $r\in[s']$ with high probability. We condition on \eqref{jl-poly} holding in what follows. We can bound the conditional probability of $I^j_a $ as follows, 
		\begin{align}
		&\Pr[I^j_a =i | I^j_1=i_1, I^j_2=i_2, \cdots I^j_{a-1}=i_{a-1}]\nonumber\\ 
		&\qquad\ge \left( 1-{1}/{5q} \right) \cdot \frac{ \| X_{i,\star} D^{a-1}P_a^\top \|_2^2}{ \|X D^{a-1} P_a^\top \|_F^2}.\label{eq:Ij-lower-bnd}
		\end{align}
		For every $a\in[q]$, line~4 of the algorithm computes $P_a=Q^q\left( X^{\otimes(q-a)} \otimes E_1^{\otimes a} \right)$, where $Q^q$ is the sketch from Lemma \ref{soda-result} with $\epsilon=\frac{1}{10q}$. By basic properties of tensor products, for every $i\in[d]$, 
		\begin{align*}
		P_a D^{a-1} X_{i,\star}^\top &= Q^q\left( X^{\otimes(q-a)} \otimes E_1^{\otimes a} \right) D^{a-1} X_{i,\star}^\top\\ 
		&= Q^q\left( \left(X^{\otimes(q-a)} D^{a-1} X_{i,\star}^\top\right) \otimes {\bf e}_1^{\otimes a} \right).
		\end{align*}
		Hence, by Lemma \ref{soda-result}, for every $a\in[q]$ and every $i\in[d]$, with high probability,
		\begin{equation}\label{tensorsketch}
		\left\| X_{i,\star} D^{a-1}P_a^\top \right\|_2^2 \in \frac{\left\| X^{\otimes(q-a)} D^{a-1} X_{i,\star}^\top \right\|_2^2}{1\pm 0.1/q}.
		\end{equation}
		By union bounding over $qd'd$ events, with high probability, \eqref{tensorsketch} holds simultaneously for all $a\in[q]$, all $j\in[d']$, and all $i\in[d]$. Therefore, conditioning on \eqref{tensorsketch} holding and using \eqref{eq:Ij-lower-bnd}, the conditional probability of $I^j_a $ satisfies 
		\begin{align}
		&\Pr[I^j_a =i | I^j_1=i_1, I^j_2=i_2, \cdots I^j_{a-1}=i_{a-1}]\nonumber\\
		&\qquad\ge \left(1-{2}/{5q}\right)  \frac{ \| X^{\otimes(q-a)} D^{a-1} X_{i,\star}^\top \|_2^2}{ \|X^{\otimes(q-a)} D^{a-1} X^\top \|_F^2}.\label{lower-bnd-Ij}
		\end{align}
		It follows from the definition of tensor products and definition of $D^a$, that
		\begin{align*}
		&\left\| X^{\otimes(q-a-1)} D^{a} X^\top \right\|_F^2\\ &= \left\| X^{\otimes(q-a-1)} D^{a-1}\cdot \text{diag}(X_{i_a,\star}) X^\top \right\|_F^2\\
		&= \left\| X^{\otimes(q-a)} D^{a-1} X_{i_a,\star}^\top \right\|_2^2
		\end{align*}
		Using this equality and inequality \eqref{lower-bnd-Ij},
		\begin{align}
		&\Pr\left[ I^j=(i_1 , i_2, \cdots i_q) \right]\nonumber\\
		&= \prod_{a=1}^{q} \Pr \left[I^j_a =i_a | I^j_1=i_1, \cdots I^j_{a-1}=i_{a-1} \right]\nonumber\\
		&\ge \prod_{a=1}^{q} \left(1-\frac{2}{5q}\right)  \frac{ \| X^{\otimes(q-a)} D^{a-1} X_{i_a,\star}^\top \|_2^2}{ \|X^{\otimes(q-a)} D^{a-1} X^\top \|_F^2}\nonumber\\
		& \ge \frac{1}{2}  \frac{ \| X^{\otimes(0)} D^{q-1} X_{i_q,\star}^\top \|_2^2}{ \|X^{\otimes(q-1)} D^{0} X^\top \|_F^2}\nonumber\\ 
		&= \frac{1}{2}  \frac{ \left| [X^{\otimes q} \cdot M ]_{(i_1,i_2,\cdots i_q),j} \right|^2}{ \|[X^{\otimes q} \cdot M]_{\star,j} \|_2^2} \label{conditional-prob-bound}
		\end{align}
		By plugging \eqref{conditional-prob-bound} back in 
		\eqref{dist-sampling-matrix-poly} we find that,
		\begin{align*}
		&\Pr\left[ S_{l,\star}= \beta^{-1/2} ({\bf e}_{i_1} \otimes {\bf e}_{i_2} \otimes \cdots {\bf e}_{i_q})^\top \right]\\
		&\qquad\ge \sum_{j=1}^{d'} \frac{1}{2} \cdot \frac{ \left| [X^{\otimes q} \cdot M ]_{(i_1,i_2,\cdots i_q),j} \right|^2}{ \|[X^{\otimes q} \cdot M]_{\star,j} \|_2^2} \cdot p_j\\
		&\qquad= \frac{1}{2}\sum_{j=1}^{d'} \frac{ \left| [X^{\otimes q} \cdot M ]_{(i_1,i_2,\cdots i_q),j} \right|^2}{ \|X^{\otimes q}  M_{\star,j} \|_2^2} \frac{ \|Q^q X^{\otimes q} M_{\star,j}\|_2^2 }{\|Q^q X^{\otimes q} M\|_F^2}\\
		&\qquad\ge \frac{1}{3} \sum_{j=1}^{d'} \frac{ \left| [X^{\otimes q} \cdot M ]_{(i_1,i_2,\cdots i_q),j} \right|^2}{ \|X^{\otimes q}  M_{\star,j} \|_2^2} \frac{ \| X^{\otimes q} M_{\star,j}\|_2^2 }{\|X^{\otimes q} M\|_F^2}\\
		&\qquad = \frac{1}{3} \frac{ \left\| [X^{\otimes q} \cdot M ]_{(i_1,i_2,\cdots i_q),\star} \right\|^2}{\|X^{\otimes q} M\|_F^2}
		\end{align*}
		Matrix $M$ is defined as $M=(B^\top B +\lambda I)^{-1/2} \cdot H$ where $H$ is a random matrix with i.i.d. Gaussian entries with $d'=C_1q\log_2n$ columns. Therefore, $H$ is a JL-transform, so for every $(i_1,i_2,\cdots i_q) \in[d]^q$, with probability $1-\frac{1}{\text{poly}(n^q)}$,
		\begin{align*}
		&\frac{\left\| [X^{\otimes q} \cdot M ]_{(i_1,i_2,\cdots i_q),\star} \right\|_2^2}{d'}\\
		&\qquad \in \left(1\pm 0.1 \right) \left\| [X^{\otimes q} ]_{(i_1,i_2,\cdots i_q),\star}  (B^\top B +\lambda I)^{-1/2}\right\|_2^2.
		\end{align*}
		Therefore, by union bounding over $d^q$ rows of $X^{\otimes q} M$, the above holds simultaneously for all $(i_1,i_2,\cdots i_q) \in[d]^q$ with high probability. Therefore,
		\begin{align*}
		&\Pr\left[ S_{l,\star}= \beta^{-1/2} ({\bf e}_{i_1} \otimes {\bf e}_{i_2} \otimes \cdots {\bf e}_{i_q})^\top \right]\\
		&\qquad\ge \frac{1}{4} \cdot \frac{ \left\| \left[X^{\otimes q} \cdot (B^\top B +\lambda I)^{-1/2} \right]_{(i_1,i_2,\cdots i_q),\star} \right\|_2^2}{\|X^{\otimes q} (B^\top B +\lambda I)^{-1/2}\|_F^2}
		\end{align*}
		Because $\frac{\beta}{s}$ is the probability of sampling row $(i_1,i_2,\cdots i_q)$ of the matrix $X^{\otimes q} (B^\top B +\lambda I)^{-1/2}$, the above inequality proves that with high probability, $S$ is a rank-$s$ row norm sampler for $X^{\otimes q} (B^\top B +\lambda I)^{-1/2}$ as in Definition \ref{def:row-samp}. 
		
		\paragraph{Runtime:} One of the expensive steps of this algorithm is the computation of $M$ in line~2 which takes $O(m^2 n + q m n \log n)$ operations since $B$ is rank $m$. Another expensive step is the computation of $P_j$ for $j=0,1,\cdots q-1$ in line~4. By Lemma \ref{soda-result}, this can be computed in time $O\left( q^6 n \log^4n + q^{5/2} \log n \cdot \text{nnz}(X) \right)$. Matrices $W_r$ for all $r\in[s']$ in line~11 of the algorithm can be computed in time $O(q^{2} \log n \cdot \text{nnz}(X))$.
		Computing the distribution $\{ p^a_r \}_{r=1}^{s'}$ in line~15 takes time $O\left(q^{11/2} s n \log^2n\right)$ for a fixed $a\in[q]$ and $l\in[s]$.  Therefore, the total time to compute this distribution for all $a$ and $l$ is $O\left(q^{13/2} s^2 n \log^2n\right)$. 
		
		The runtime to compute the distribution $\{ q^q_i \}_{i\in h^{-1}(t)}$ in line~17 depends on the sparsity of $X_{h^{-1}(t),\star}$, i.e.,  $\text{nnz}(X_{h^{-1}(t),\star})$. To bound the sparsity of $X_{h^{-1}(t),\star}$, note that, $\text{nnz}(X_{h^{-1}(t),\star}) = \sum_{i=1}^d \mathbf{1}_{\{i \in h^{-1}(t)\}} \cdot \text{nnz}(X_{i,\star})$.
		Let us introduce the random variables $S_1,S_2,\cdots S_d$ defined as $S_i = \left(\mathbf{1}_{\{i \in h^{-1}(t)\}} - \frac{1}{s'}\right) \cdot \text{nnz}(X_{i,\star})$ for $i\in[d]$.
		Since the hash function $h$ is fully independent, the random variables $S_1,S_2,\cdots S_d$ are independent. Also, each of these random variables is zero mean and uniformly bounded, i.e., $\EE[S_i]=0$  and $|S_i|\le n$ for each $i\in [d]$. Therefore we can invoke Bernstein's inequality (Appendix A). Let $Z=\sum_{i=1}^d S_i$. Then the variance of the sum is bounded as
		$\sum_{i=1}^d \EE[S_i^2] \le \sum_{i=1}^d \frac{1}{s'} \text{nnz}(X_{i,\star})^2 \le \frac{n}{s'} \cdot \text{nnz}(X)$.
		
		By invoking Bernstein's inequality, for some constant $C$, $\Pr \left[ |Z| \ge C\log_2 n \cdot \left( \sqrt{\frac{n}{s'} \cdot \text{nnz}(X)} + n \right) \right] \le \frac{1}{\text{poly}(n)}$.
		Hence, for every $t\in[s']$,  with high probability $\text{nnz}(X_{h^{-1}(t),\star}) = O\left( \left( {\text{nnz}(X)/s'} + n \right) \log n \right)$. By union bounding over $s'$ events, with high probability, $\text{nnz}(X_{h^{-1}(t),\star}) = O\left( \left( {\text{nnz}(X)/s'} + n \right) \log n \right)$, simultaneously for all $t\in[s']$
		which implies that the distribution $\{ q^q_i \}_{i\in h^{-1}(t)}$ in line~17 of the algorithm can be computed in time $O\left( q^2 n \log^2 n + q^2 \log^2 n \cdot {\text{nnz}(X)/s'} \right)$ for a fixed $a \in [q]$ and a fixed $l \in [s]$.  Therefore the total time to compute this distribution for all $a$ and all $l$ is $O\left( q^3 s n \log^2 n + q^{3/2} \log^2 n \cdot {\text{nnz}(X)} \right)$. 
		
		Finally the last expensive step is the computation of quantities $p^*_b$ and $q^*_b$ in lines~25 and 26 of the algorithm. Both of these quantities can be computed in time $O\left( q^{11/2} s n \log^2 n + q^2 \log^2n \cdot \text{nnz}(X)/s' \right)$ for a fixed $j \in [d']$ and a fixed $b\in[q]$. Therefore the total time to compute these quantities for all $l$, all $j$, and all $b$ is $O\left( q^{15/2} s^2 n \log^3 n + {q^{5/2}} \log^3n \cdot \text{nnz}(X) \right)$.
		Therefore the total runtime of Algorithm~\ref{alg:rotatedrowsampler-poly} is $O\left( m^2 n + q^{15/2} s^2 n \log^3 n + q^{5/2} \log^3n \cdot \text{nnz}(X) \right)$.
	\end{proof}
	We prove Theorem \ref{thm:poly} in Appendix E.
	
	\subsection{Adaptive Sampling for the Gaussian Kernel}
	Consider the lifting corresponding to the Gaussian kernel, $k(x,y)=e^{-\|x-y\|_2^2/2}$, that can be obtained through a Taylor expansion. This feature mapping was exploited in \cite{ahle2019oblivious} to obtain an efficient subspace embedding for the Gaussian kernel via sketching the polynomial terms in its Taylor expansion. For datasets with bounded radius, the Gaussian kernel can be well-approximated by a superposition of low-degree polynomial kernels. We formally define this approximate feature mapping (lifting) as follows.
	
	\begin{definition}[Polynomial Lifting for Gaussian Kernel]\label{def:poly-feature-Gauss}
		For any integer $q$ the \emph{degree-$q$ polynomial lifting for Gaussian kernel} is the mapping $\phi_q:\RR^d \rightarrow \RR^D$, defined as,
		\[	\phi_q(x) = e^{-\|x\|_2^2/2}
		\left(\frac{x^{\otimes 0}}{\sqrt{0!}} \oplus \frac{x^{\otimes 1}}{\sqrt{1!}} \oplus \frac{x^{\otimes 2}}{\sqrt{2!}} \oplus \cdots \frac{x^{\otimes q}}{\sqrt{q!}}\right),\]
		for $x\in\RR^d$, where $D=\sum_{j=0}^q d^j$.
	\end{definition}
	
	\begin{claim}\label{claim:poly-lift-gauss}
		Let $x_1, x_2, \cdots x_n \in\RR^d$ be a dataset with bounded radius, i.e.,  $\|x_i\|_2^2\le r$ for all $i\in[n]$. Suppose $K\in\RR^{n\times n}$ is the Gaussian kernel corresponding to this dataset ($K_{i,j} = e^{-\|x_i-x_j\|_2^2/2}$). Also suppose that $\phi_q$ is the degree-$q$ polynomial lifting for the Gaussian kernel as in Definition \ref{def:poly-feature-Gauss}. If $A$ is a matrix with $n$ columns whose columns are obtained by applying the map $\phi_q$ on the data points, i.e., $A_{\star,i} = \phi_q(x_i)$, then as long as $q=\Omega\left(r+\log n \right)$, we have $\|A^\top A - K\|_{op} \le \frac{1}{\text{poly}(n)}$.
	\end{claim}
	
	Therefore, to find a spectral approximation to the Gaussian kernel $K$ for bounded datasets, it is enough to find a spectral approximation to $A^\top A$, where $A$ is the matrix defined in the above claim. 
	We have designed an efficient adaptive sampling method for tensor products of the form $X^{\otimes j}$ in the previous section. Since matrix $A$ is a concatenation of tensor products $X^{\otimes j}$ for $j=0,1, \cdots q$, using our iterative leverage score sampling procedure for the polynomial kernel we can spectrally approximate $A^\top A$ in nearly $\text{nnz}(X)$ time. We present a full algorithm which can perform recursive leverage score sampling on matrix $A$ and analyze it in Appendix F and prove Theorem~\ref{thm:gauss} in Appendix G.
	
	\subsection{Generalization to dot-product Kernels}\label{sec:dot-prod}
	An important technical contribution of this paper is a sampling method that can embed the polynomial kernel using near-optimal runtime. Additionally, our method can be used for embedding a wide class of kernels that can be well-approximated by low-degree polynomials. In particular, our sampling method can be applied to any dot-product kernel with a rapidly convergent Taylor expansion. In this section, we argue how our method can be generalized to such kernels.
	
	The underlying observation that enables us to extend our subspace embedding to the class of dot-product kernels is a classical result in harmonic analysis due to \citet{schoenberg1988positive}, that characterizes positive definite functions in a Hilbert space. This observation is simply the fact that any dot-product kernel $k:\RR^d \times \RR^d \to \RR$ defined as $k(x,y) = f(\langle x, y \rangle)$ must have a Taylor expansion with only non-negative coefficients, i.e., $k$ is a kernel function if and only if $f(\alpha) = \sum_{j=0}^\infty a_j \alpha^j$, $a_j \ge 0$ for all $j\in \mathbb{Z}$. As a result, truncating this sum at any point results in a valid kernel, that is $k_q(x,y):=\sum_{j=0}^q a_j \langle x, y \rangle^j$ is a valid positive definite kernel. 
	
	For most dot-product kernels used in practice, the coefficients $a_j$ decay at least exponentially. If this is the case, then $|k_q(x,y) - k(x,y)| \le \frac{1}{\text{poly}(n)}$ for any $x,y \in \RR^d$ with $\|x\|_2^2 , \|y\|_2^2 \le r$ and $q = \Omega(r\log n)$. Hence, in order to obtain a subspace embedding for kernel $k$ on any dataset with bounded $\ell_2$ radius, it is enough to find a subspace embedding for the truncated kernel $k_q(x,y) = \sum_{j=0}^q a_j \langle x, y \rangle^j$. Since this kernel is a superposition of polynomial kernels, we can apply our subspace embedding for the polynomial kernel from Section~\ref{sec:polynomial} to each of the polynomial terms. This will result in a near input sparsity time subspace embedding for any dot-product kernel whose Taylor expansion decays at least exponentially.
	
	An example of a well known dot product kernel is the inverse polynomial kernel defined as $k(x,y) = \frac{1}{2 - \langle x,y \rangle}$. The Taylor expansion of this kernel is $k(x,y) = \sum_{j=0}^\infty 2^{-j-1} \langle x,y \rangle^j$. Therefore, if we let $q = \Theta(\log n)$ then for any $x,y \in \RR^d$ with $\|x\|_2^2 , \|y\|_2^2 \le 1$, $|k_q(x,y) - k(x,y)| \le \frac{1}{\text{poly}(n)}$, where $k_q(x,y) = \sum_{j=0}^q 2^{-j-1} \langle x,y \rangle^j$. Hence, we can obtain a subspace embedding for the inverse polynomial kernel in nearly $\text{nnz}(X)$ time by applying our sampling method from Section~\ref{sec:polynomial} to polynomials of degree $O(\log n)$ in this Taylor expansion.
	
	\section{Experiments}
	In this section we assess the performance of our result for embedding the Gaussian kernel (Theorem~\ref{thm:gauss}) against the Fourier features ({\bf FF}) method \cite{rahimi2008random}, {\bf Nystrom} method \cite{musco2017recursive}, as well as the {\bf Oblivious} sketching method of \cite{ahle2019oblivious}. The results are summarized in Table \ref{experiment-table}\footnote{We repeated the experiments with 5 different random seeds and reported the average RMSE and runtime in Table~\ref{experiment-table}.}. Our importance sampling algorithm is a recursive procedure given in Algorithm \ref{alg:outerloop}. In this set of experiments, we also consider a variant of our sampling algorithm that runs only a single round of the recursive sampling and hence is considerably faster. This variant is equivalent to sampling rows of the lifting matrix $\Phi$ with probabilities proportional to the squared row norms. We denote this variant of our method by {\bf Row norm} and denote the full recursive importance sampling algorithm by {\bf Adaptive}. 
	The target dimension of all methods is denoted by $s$ in Table \ref{experiment-table}.
	
	We base our comparison on the four standard large-scale regression datasets evaluated in \cite{le2013fastfood}. The size of the data points is denoted by $n$ and the dimensionality is denoted by $d$ in Table \ref{experiment-table}. In all experiments, we first find a low-rank approximation to the kernel matrix using various feature sampling/sketching techniques. Then, using the kernel’s proxy, we find an approximate regressor by solving an $\ell_2$ regularized least-squares problem. For all methods, Table~\ref{experiment-table} reports the total time to train the regressors, including the runtime of feature sampling and the runtime of linear regression. We use the same hyperparameters (kernel bandwidth and regularization parameter) across all kernel approximation methods which were selected via cross-validation on the Fourier features method, as our baseline method. For every method, we set the number $s$ of features to the smallest value such that increasing the number of features does not improve the error non-negligibly.
	
	The \emph{Row norm} variant of our method is as fast as the \emph{FF} method and runs significantly faster than the \emph{Nystrom} and \emph{Oblivious} methods while having superior testing RMSE. Our  full algorithm, \emph{Adaptive}, has even better performance than our single round variant \emph{Row norm} in terms of RMSE on the test set and achieves a better RMSE while having a significantly smaller target dimension $s$ than all other methods. In terms of runtime, our full \emph{Adaptive} method is no worse than \emph{Nystrom} but is slower than our single round \emph{Row norm} method. Our \emph{Adaptive} method has a slightly better RMSE than the \emph{Oblivious} method and runs slower, but it achieves a significantly smaller target dimension $s$. However, our single round \emph{Row norm} variant is significantly faster than \emph{Oblivious}.
	
	\begin{table}[t]
		\caption{The RMSE on the test set along with the total training time of approximate KRR via various approximation methods.}
		\label{experiment-table}
		\vskip 0.1in
		\begin{center}
			\begin{small}
				\begin{sc}
					\scalebox{0.76}{
						\begin{tabular}{lcccr}
							\toprule
							Dataset:& Wine  & Insurance & CT location & Forest \\
							& $n=6,497$ & $n=9,822$ & $n=53,500$ & $n=581,012$ \\
							& $d=11$ & $d=85$ & $d=384$ & $d=54$ \\
							
							\midrule
							{\bf FF}   & {$0.736$, $2$ sec} & {$0.231$}, $1$ sec & {$3.89$}, $1$ min & {$1.00$}, $3$ min \\
							& $s=5000$ & $s=2000$ & $s=4000$ & $s=1000$ \\
							\midrule
							{\bf Nystrom}   & $0.730$, $1.5$ min & $0.231$, $1.5$ min & $3.86$, $8.5$ min & $1.03$, $8$ min \\
							& $s=2000$ & $s=2000$ & $s=1500$ & $s=500$ \\
							\midrule
							{\bf Oblivious}   & $0.732$, $13$ sec & $0.231$, $20$ sec & $3.70$, $3.5$ min & $1.05$, $2.5$ min \\
							& $s=1024$ & $s=1024$ & $s=5120$ & $s=320$ \\
							\midrule
							{\bf Row norm} & $0.727$, $3$ sec & $0.231$, $2$ sec & $3.68$, $1$ min & $1.08$, $2.5$ min \\
							& $s=5000$ & $s=1500$ & $s=6000$ & $s=1000$ \\
							\midrule
							{\bf Adaptive} & $0.723$, $15$ sec & $0.232$, $6$ sec & $3.72$, $8.5$ min & $1.05$, $7$ min \\
							& $s=400$ & $s=400$ & $s=2800$ & $s=500$ \\
							\bottomrule
						\end{tabular}
					}
				\end{sc}
			\end{small}
		\end{center}
		\vskip -0.3in
	\end{table}
	
	\section*{Acknowledgements}
	D.\@ P.\@ Woodruff was supported in part by Office of Naval Research (ONR) grant N00014-18-1-2562. 
	
	\bibliographystyle{icml2020}
	
	\bibliography{example_paper}

\begin{thebibliography}{20}
\providecommand{\natexlab}[1]{#1}
\providecommand{\url}[1]{\texttt{#1}}
\expandafter\ifx\csname urlstyle\endcsname\relax
  \providecommand{\doi}[1]{doi: #1}\else
  \providecommand{\doi}{doi: \begingroup \urlstyle{rm}\Url}\fi

\bibitem[Ahle et~al.(2020)Ahle, Kapralov, Knudsen, Pagh, Velingker, Woodruff,
  and Zandieh]{ahle2019oblivious}
Ahle, T.~D., Kapralov, M., Knudsen, J.~B., Pagh, R., Velingker, A., Woodruff,
  D.~P., and Zandieh, A.
\newblock Oblivious sketching of high-degree polynomial kernels.
\newblock In \emph{Proceedings of the Fourteenth Annual ACM-SIAM Symposium on
  Discrete Algorithms}, pp.\  141--160. SIAM, 2020.

\bibitem[Ailon \& Chazelle(2006)Ailon and Chazelle]{ailon2006approximate}
Ailon, N. and Chazelle, B.
\newblock Approximate nearest neighbors and the fast johnson-lindenstrauss
  transform.
\newblock In \emph{Proceedings of the thirty-eighth annual ACM symposium on
  Theory of computing}, pp.\  557--563, 2006.

\bibitem[Alaoui \& Mahoney(2015)Alaoui and Mahoney]{alaoui2015fast}
Alaoui, A. and Mahoney, M.~W.
\newblock Fast randomized kernel ridge regression with statistical guarantees.
\newblock In \emph{Advances in Neural Information Processing Systems}, pp.\
  775--783, 2015.

\bibitem[Avron et~al.(2014)Avron, Nguyen, and Woodruff]{avron2014subspace}
Avron, H., Nguyen, H., and Woodruff, D.
\newblock Subspace embeddings for the polynomial kernel.
\newblock In \emph{Advances in neural information processing systems}, pp.\
  2258--2266, 2014.

\bibitem[Avron et~al.(2017{\natexlab{a}})Avron, Clarkson, and
  Woodruff]{avron2017faster}
Avron, H., Clarkson, K.~L., and Woodruff, D.~P.
\newblock Faster kernel ridge regression using sketching and preconditioning.
\newblock \emph{SIAM Journal on Matrix Analysis and Applications}, 38\penalty0
  (4):\penalty0 1116--1138, 2017{\natexlab{a}}.

\bibitem[Avron et~al.(2017{\natexlab{b}})Avron, Kapralov, Musco, Musco,
  Velingker, and Zandieh]{avron2017random}
Avron, H., Kapralov, M., Musco, C., Musco, C., Velingker, A., and Zandieh, A.
\newblock Random fourier features for kernel ridge regression: Approximation
  bounds and statistical guarantees.
\newblock In \emph{Proceedings of the 34th International Conference on Machine
  Learning-Volume 70}, pp.\  253--262. JMLR. org, 2017{\natexlab{b}}.

\bibitem[Bach(2013)]{bach2013sharp}
Bach, F.
\newblock Sharp analysis of low-rank kernel matrix approximations.
\newblock In \emph{Conference on Learning Theory}, pp.\  185--209, 2013.

\bibitem[Boucheron et~al.(2013)Boucheron, Lugosi, and
  Massart]{boucheron2013concentration}
Boucheron, S., Lugosi, G., and Massart, P.
\newblock \emph{Concentration inequalities: A nonasymptotic theory of
  independence}.
\newblock Oxford university press, 2013.

\bibitem[Cohen et~al.(2016)Cohen, Musco, and Pachocki]{cohen2016online}
Cohen, M.~B., Musco, C., and Pachocki, J.
\newblock Online row sampling.
\newblock \emph{arXiv preprint arXiv:1604.05448}, 2016.

\bibitem[Cohen et~al.(2017)Cohen, Musco, and Musco]{cohen2017input}
Cohen, M.~B., Musco, C., and Musco, C.
\newblock Input sparsity time low-rank approximation via ridge leverage score
  sampling.
\newblock In \emph{Proceedings of the Twenty-Eighth Annual ACM-SIAM Symposium
  on Discrete Algorithms}, pp.\  1758--1777. SIAM, 2017.

\bibitem[Dasgupta \& Gupta(2003)Dasgupta and Gupta]{dasgupta2003elementary}
Dasgupta, S. and Gupta, A.
\newblock An elementary proof of a theorem of johnson and lindenstrauss.
\newblock \emph{Random Structures {\&} Algorithms}, 22\penalty0 (1):\penalty0
  60--65, 2003.

\bibitem[Kane \& Nelson(2014)Kane and Nelson]{kane2014sparser}
Kane, D.~M. and Nelson, J.
\newblock Sparser johnson-lindenstrauss transforms.
\newblock \emph{Journal of the ACM (JACM)}, 61\penalty0 (1):\penalty0 4, 2014.

\bibitem[Kapralov et~al.(2014)Kapralov, Lee, Musco, Musco, and
  Sidford]{kapralov2014single}
Kapralov, M., Lee, Y.~T., Musco, C., Musco, C., and Sidford, A.
\newblock Single pass spectral sparsification in dynamic streams.
\newblock In \emph{2014 IEEE 55th Annual Symposium on Foundations of Computer
  Science}, pp.\  561--570. IEEE, 2014.

\bibitem[Le et~al.(2013)Le, Sarl{\'o}s, and Smola]{le2013fastfood}
Le, Q., Sarl{\'o}s, T., and Smola, A.
\newblock Fastfood-approximating kernel expansions in loglinear time.
\newblock In \emph{Proceedings of the international conference on machine
  learning}, volume~85, 2013.

\bibitem[Musco \& Musco(2017)Musco and Musco]{musco2017recursive}
Musco, C. and Musco, C.
\newblock Recursive sampling for the nystrom method.
\newblock In \emph{Advances in Neural Information Processing Systems}, pp.\
  3833--3845, 2017.

\bibitem[Pham \& Pagh(2013)Pham and Pagh]{pham2013fast}
Pham, N. and Pagh, R.
\newblock Fast and scalable polynomial kernels via explicit feature maps.
\newblock In \emph{Proceedings of the 19th ACM SIGKDD international conference
  on Knowledge discovery and data mining}, pp.\  239--247, 2013.

\bibitem[Rahimi \& Recht(2008)Rahimi and Recht]{rahimi2008random}
Rahimi, A. and Recht, B.
\newblock Random features for large-scale kernel machines.
\newblock In \emph{Advances in neural information processing systems}, pp.\
  1177--1184, 2008.

\bibitem[Schoenberg(1988)]{schoenberg1988positive}
Schoenberg, I.
\newblock Positive definite functions on spheres.
\newblock \emph{Duke Math. J}, 1:\penalty0 172, 1988.

\bibitem[Tropp(2011)]{tropp2011improved}
Tropp, J.~A.
\newblock Improved analysis of the subsampled randomized hadamard transform.
\newblock \emph{Advances in Adaptive Data Analysis}, 3\penalty0
  (01n02):\penalty0 115--126, 2011.

\bibitem[Zandieh et~al.(2020)Zandieh, Nouri, Velingker, Kapralov, and
  Razenshteyn]{pmlr-v108-zandieh20a}
Zandieh, A., Nouri, N., Velingker, A., Kapralov, M., and Razenshteyn, I.
\newblock Scaling up kernel ridge regression via locality sensitive hashing.
\newblock In \emph{Proceedings of the Twenty Third International Conference on
  Artificial Intelligence and Statistics}, volume 108 of \emph{Proceedings of
  Machine Learning Research}, pp.\  4088--4097, Online, 26--28 Aug 2020. PMLR.

\end{thebibliography}

	\appendix
\section{Bernstein's Inequality}
We use Bernstein's concentration inequality given in the following lemma.
\begin{lemma}\label{claim:bernstein}
	Let $S_1,S_2, \cdots S_n$ be independent, mean-$0$, real-valued random variables, and assume that each one is uniformly bounded:
	\[\EE[S_k]=0 \text{ and } |S_k|\le L \text{ for each }k=1,2, \cdots n\]
	Let $Z =\sum_{k=1}^n S_k$ , and let $v$ denote the variance of the sum:
	\[ v= \EE[Z^2] = \sum_{k=1}^n \EE[S_k^2]. \]
	Then,
	\[ \Pr[|Z| \ge t] \le 2 e^{\frac{-t^2/2}{v+Lt/3}}. \]
\end{lemma}
See \cite{boucheron2013concentration} for a proof of this result.

\section{Properties of Leverage Scores}

In this section we present the definition and basic properties of the ridge leverage scores of a matrix $\Phi\in \RR^{D\times n}$. For every regularization parameter $\lambda>0$ and every $i \in [D]$ the ridge leverage score of the $i^{th}$ row of $\Phi$ is defined as,
\[l^{\lambda}_i \equiv  \phi_i^\top (\Phi^\top \Phi+\lambda I)^{-1}\phi_i,\]
where $\phi_i \in \RR^{n}$ is the $i^{th}$ row of $\Phi$, treated as a column vector.
There is a connection between the ridge leverage scores of $\Phi$ and the statistical dimension of $\Phi^\top \Phi$. The sum of the ridge leverage scores is equal to the statistical dimension of the kernel matrix $K=\Phi^\top \Phi$,
\[s_\lambda \equiv {\bf tr}\left(\Phi^\top \Phi(\Phi^\top \Phi+\lambda I)^{-1} \right)=\sum_{i \in [d]} l^\lambda_i.\]
We next present a lemma which shows that ridge leverage score sampling is an optimal sampling strategy for achieving the spectral guarantee of \eqref{spectral-bound} (up to an  $O(\log n)$ factor),
\begin{lemma}\label{lem:lss}
	Let $\Phi$ be a $D\times n$ matrix with rows $\phi_1,\phi_2, \cdots \phi_D$ and with ridge leverage scores $l_i^\lambda=\phi_i^\top (\Phi^\top \Phi + \lambda I)^{-1}\phi_i$ for all $i\in [D]$. Let $\epsilon,\lambda>0$. Assume that we are given a probability distribution $\{p_i\}_{i=1}^D$ such that $p_i \ge \alpha \cdot \frac{l^\lambda_i}{\sum_{j \in [D]}l_j^\lambda}$ for every $i\in[D]$ and some $\alpha\in(0,1)$. Construct the sampling matrix $\Pi\in\RR^{s \times D}$ by generating $s$ i.i.d. samples $j_1, j_2, \cdots j_s \in [D]$ with distribution $\{p_i\}_{i=1}^D$ and letting the $r^{th}$ row of $\Pi$ be $\frac{1}{\sqrt{sp_{j_r}}}{\bf e}_{j_r}^\top$ for every $r\in [s]$, where ${\bf e}_{1}, {\bf e}_{2}, \cdots {\bf e}_{D}\in \RR^D$ are the standard basis vectors. If the number of rows of $\Pi$ is at least $s\ge \frac{4\log_2n}{\alpha\epsilon^2}\cdot \sum_{j \in [D]}l_j^\lambda$, then with high probability,
	\[ \frac{\Phi^\top \Phi + \lambda I}{1+\epsilon} \preceq \Phi^\top \Pi^\top \Pi \Phi + \lambda I \preceq \frac{\Phi^\top \Phi + \lambda I}{1-\epsilon}.\]
\end{lemma}
\begin{proof}
	This guarantee for leverage score sampling is well-known. See, for example, \cite{cohen2016online,cohen2017input}.
\end{proof}

\section{Proof of Lemma \ref{resursive-rlss-lem}}
Let $S_t$ be the sampling matrix and let $\lambda_t$ be the regularizing parameter in the $t^{th}$ iteration of Algorithm \ref{alg:outerloop}.
The proof of the lemma proceeds by induction. We define the event $\mathcal{E}_t$ as the set of all sampling matrices $S_t$ that satisfy the following condition,
\[	\frac{\Phi^\top \Phi + \lambda_t I}{1+\epsilon} \preceq \Phi^\top S_t^\top S_t \Phi + \lambda_t I \preceq \frac{\Phi^\top \Phi + \lambda_t I}{1-\epsilon}.\]
We show by induction that for all $t=0, 1, \cdots T$, the invariant $\mathcal{E}_t$ conditionally holds with high probability, that is,
\[\Pr[\mathcal{E}_t | \mathcal{E}_{t-1}] \ge 1-\frac{1}{\text{poly}(n)}.\]
The base of the induction corresponds to $t=0$.
For $t=0$ we have that $S_0\Phi = 0$ and $\lambda_0 = \frac{\|\Phi\|_F^2}{\epsilon}$, and therefore, $\Phi^\top \Phi \preceq \epsilon{\lambda_0} I$, which implies that,
\[\frac{\Phi^\top \Phi + \lambda_0 I}{1+\epsilon} \preceq \Phi^\top S_0^\top S_0\Phi + \lambda_0 I \preceq \Phi^\top \Phi + \lambda_0 I.\]
Therefore, $\Pr[\mathcal{E}_0] =1$, which proves the {\bf base case of the induction}.

Now to prove the inductive step, note that conditioned on the event $\mathcal{E}_{t}$ holding for some $t \ge0$, we find that
$$\frac{\Phi^\top \Phi + \lambda_{t} I}{1+\epsilon} \preceq \Phi^\top S_{t}^\top S_{t} \Phi + \lambda_{t} I \preceq \frac{\Phi^\top \Phi + \lambda_{t} I}{1-\epsilon}.$$
By definition of ridge leverage scores, $l_i^{\lambda_{t+1}} = \phi_i^\top (\Phi^\top \Phi + \lambda_{t+1} I)^{-1} \phi_i$, and noting that $\lambda_{t+1} = \lambda_{t}/2$, we have
$$\phi_i^\top (\Phi^\top \Phi + \lambda_{t} I)^{-1} \phi_i \le l_i^{\lambda_{t+1}} \le 2 \phi_i^\top (\Phi^\top \Phi + \lambda_{t} I)^{-1} \phi_i.$$
By the inductive hypothesis, for $\epsilon\le \frac{1}{3}$, we have,
\[	\frac{1}{3} \cdot l_i^{\lambda_{t+1}} \le \phi_i^\top \left(\Phi^\top S_{t}^\top S_{t} \Phi + \lambda_{t} I\right)^{-1} \phi_i \le \frac{4}{3} \cdot l_i^{\lambda_{t+1}} .\]
Now note that by the assumption of the lemma, $S_{t+1}$ is a rank-$s$ row norm sampler for the matrix $\Phi(\Phi^\top S_{t}^\top S_{t} \Phi + \lambda_{t} I)^{-1/2}$. Therefore there exists a probability distribution $\{p_i\}_{i=1}^D$ such that $S_{t+1}\in \RR^{s'\times D}$ is the corresponding sampling matrix to this probability distribution constructed as in Definition \ref{def:row-samp}. This probability distribution satisfies, 
\begin{align*}
p_i &\ge \frac{1}{4}\frac{\phi_i^\top \left(\Phi^\top S_{t}^\top S_{t} \Phi + \lambda_{t} I\right)^{-1} \phi_i}{\sum_{j\in[D]} \phi_j^\top \left(\Phi^\top S_{t}^\top S_{t} \Phi + \lambda_{t} I\right)^{-1} \phi_j}\\
&\ge \frac{1}{16}\cdot \frac{l_i^{\lambda_{t+1}}}{\sum_{j\in[D]} l_j^{\lambda_{t+1}}}.
\end{align*}
Therefore because $s=C\frac{\mu}{\epsilon^2} \log_2n \ge C \frac{s_{\lambda}}{\epsilon^2} \log_2n \ge C \frac{s_{\lambda_{t+1}}}{\epsilon^2} \log_2n$, if $C$ is a large enough constant, by Lemma~\ref{lem:lss}, 
\[\Pr[\mathcal{E}_{t+1} | \mathcal{E}_{t}] \ge 1-\frac{1}{\text{poly}(n)}.\]
This completes the {\bf inductive step}.
By union bounding over all $t$, we get that,
$$\Pr[\mathcal{E}_T] \ge 1 - \frac{1}{\text{poly}(n)}.$$ 
Hence, since $\lambda_T \le \lambda$, with high probability the following holds for the sampling matrix $\Pi = S_T$,
\[ \frac{\Phi^\top \Phi + \lambda I}{1+\epsilon} \preceq \Phi^\top \Pi^\top \Pi \Phi + \lambda I \preceq \frac{\Phi^\top \Phi + \lambda I}{1-\epsilon}.\]
This completes the proof of the lemma.

\section{Proof of Lemma \ref{soda-result}}
By invoking Theorem 1.2 of \cite{ahle2019oblivious}, there exists a sketch $S^q \in\RR^{s \times d^q}$ such that if $s=\Omega\left(\frac{q}{\epsilon^2} \log^3 \frac{1}{\delta}\right)$ then for any $y\in \RR^{d^q}$,
\[ \Pr\left[ \|S^q y \|_2^2 \in (1\pm \epsilon) \|y\|_2^2\right] \ge 1- \delta/2. \]
Let $G \in \RR^{m \times s}$ be a random matrix with i.i.d. normal entries. Thus, $G$ is a JL transform with high probability. By the analysis in \cite{kane2014sparser}, if $m = \Omega \left( \frac{1}{\epsilon^2} \log \frac{1}{\delta} \right)$ then for any $z \in \RR^{s}$,
\[ \Pr\left[ \|G z \|_2^2 \in (1\pm \epsilon) \|z\|_2^2\right] \ge 1- \delta/2. \]
Therefore if we let $Q^q := G S^q$ then we have that this matrix is of size $m \times d^q$ and also by a union bound, 
for any $y\in \RR^{d^q}$,
\[ \Pr\left[ \|Q^q y \|_2^2 \in (1\pm \epsilon) \|y\|_2^2\right] \ge 1- \delta. \]

\paragraph{Runtime:} As shown in \cite{ahle2019oblivious}, the sketch $S^q$ can be applied to $v_1 \otimes v_2 \otimes \cdots v_q$ by a recursive application of $O(q)$ independent instances of OSNAP and SRHT sketches on the vectors $v_i$ and their sketched versions. The sketch $S^q$ in \cite{ahle2019oblivious} can be represented by a binary tree with $q$ leaves where the leaves are OSNAP sketches and the internal nodes are SRHT sketches. Therefore, by Theorem 1.2 of \cite{ahle2019oblivious}, $S^q x^{\otimes q}$ can be computed in time $O\left( qs \log s + \frac{q^{3/2}}{\epsilon} \log \frac{1}{\delta} \cdot \text{nnz}(x) \right)$. From the binary tree structure of the sketch it follows that after computing $S^q x^{\otimes q}$, $S^q \left(x^{\otimes q-1} \otimes {\bf e}_1\right)$ can be computed by updating the path from one of the leaves to the root of the binary tree which amounts to applying one OSNAP transform on ${\bf e}_1$ and applying $O(\log q)$ instances of SRHT on the intermediate vectors which can be computed in a total extra time of $O( s\log s \log q )$. By this argument, it follows that $S^q \left( x^{\otimes q-j} \otimes {\bf e}_1^{j} \right)$ can be computed for all $j=0,1,2, \cdots q$ in total time $O\left( q s \log q \log s + \frac{q^{3/2}}{\epsilon} \log\frac{1}{\delta} \cdot \text{nnz}(x) \right)$. By choosing a large enough $s = O\left( \frac{q}{\epsilon^2} \log^3 \frac{1}{\delta} \right)$, this runtime will be $O\left( \frac{q^2 \log^2 \frac{q}{\epsilon}}{\epsilon^2} \log^3 \frac{1}{\delta} + \frac{q^{3/2}}{\epsilon} \log\frac{1}{\delta} \cdot \text{nnz}(x) \right)$.
Also, the time to apply the Gaussian sketch $G$, with large enough target dimension $m = O\left( \frac{1}{\epsilon^2} \log \frac{1}{\delta} \right)$,  to any $s$-dimensional vector is $O\left( \frac{s}{\epsilon^2} \log \frac{1}{\delta} \right)$. Hence the total time to compute $Q^q \left( x^{\otimes q-j} \otimes {\bf e}_1^{j} \right)$ for all $j=0,1,2, \cdots q$ is $O\left( \frac{q^2 }{\epsilon^4} \log^4 \frac{1}{\delta} + \frac{q^{3/2}}{\epsilon} \log\frac{1}{\delta} \cdot \text{nnz}(x) \right)$. 

\section{Proof of Theorem \ref{thm:poly}}
We run the recursive leverage score sampling procedure of Algorithm~\ref{alg:outerloop} on the feature matrix $\Phi = X^{\otimes q}$ with $\mu=O(s_\lambda)$. Each time Algorithm \ref{alg:outerloop} invokes the procedure \textsc{RowSampler}, we run Algorithm~\ref{alg:rotatedrowsampler-poly}. By Lemma \ref{lem:rotatedrowsampler-poly}, for any $\lambda'>0$, any integers $q,s'$, and any matrices $X, B$, with high probability, the procedure \textsc{RowSampler}$(X,q,B,\lambda',s')$ of Algorithm \ref{alg:rotatedrowsampler-poly} outputs a rank-$s$ row norm sampler for matrix $X^{\otimes q}(B^\top B + \lambda' I)^{-1/2} = \Phi (B^\top B + \lambda' I)^{-1/2}$. Therefore, since the total number of times Algorithm \ref{alg:rotatedrowsampler-poly} is invoked by Algorithm \ref{alg:outerloop} is $O\left(\log \frac{ \| X^{\otimes q} \|_F^2 }{\epsilon \lambda} \right) = O\left(\log \frac{{\bf tr}(K)}{\epsilon \lambda} \right)=O(\log n)$, by a union bound, with high probability the preconditions of Lemma \ref{resursive-rlss-lem} hold and hence we can invoke this lemma to conclude that the sampler $\Pi$ that Algorithm \ref{alg:outerloop} outputs satisfies the following with high probability,
\[ \frac{\Phi^\top \Phi + \lambda I}{1+\epsilon} \preceq \Phi^\top \Pi^\top \Pi \Phi + \lambda I \preceq \frac{\Phi^\top \Phi + \lambda I}{1-\epsilon}.\] 
Therefore, if we let $Z=\Pi \Phi$, the theorem follows because $\Pi$ has $s=O\left( \frac{s_\lambda}{\epsilon^2} \log n\right)$ rows. Moreover, the primitive \textsc{RowSampler}$(X,q,B,\lambda',s')$ of Algorithm \ref{alg:rotatedrowsampler-poly} is invoked $O\left(\log \frac{ \| X^{\otimes q} \|_F^2 }{\epsilon \lambda} \right) = O\left(\log \frac{{\bf tr}(K)}{\epsilon \lambda} \right) = O(\log n)$ times with inputs $s' = O(\frac{s_\lambda}{\epsilon^2} \log n)$ and a matrix $B$ with $O\left(\frac{s_\lambda}{\epsilon^2} \log n\right)$ rows. Each invocation, by Lemma \ref{lem:rotatedrowsampler-poly}, takes $O\left( {\text{poly}(\epsilon^{-1},q,\log n) \cdot s_\lambda^2n} + q^{5/2} \log^3n \cdot \text{nnz}(X) \right)$ operations. Hence the total runtime of the algorithm is $O\left( {\text{poly}(\epsilon^{-1},q,\log n) \cdot s_\lambda^2n} + q^{5/2} \log^4n \cdot \text{nnz}(X) \right)$.

\section{\textsc{RowSampler} for the  Gaussian Kernel}
We design a procedure \textsc{RowSampler} that takes in the dataset matrix $X \in\RR^{d\times n}$ together with an $m\times n$ matrix $B$ and performs row norm sampling (see Definition \ref{def:row-samp}) on matrix $\phi_q(X) (B^\top B + \lambda I)^{-1/2}$, where $\phi_q(X)$ is a matrix with $n$ columns which are obtained by applying the mapping $\phi_q$ of Definition \ref{def:poly-feature-Gauss} on each of the columns of $X$, i.e., $[\phi_q(X)]_{\star,i} = \phi_q(X_{\star,i})$. Algorithm \ref{alg:rowsampler-gauss} performs this task.
\begin{algorithm}[h!]
	\caption{\algoname{RowSampler for Gaussian Kernel}}
	{\bf input}: $X \in \RR^{d \times n}$, $q\in \mathbb{Z}$, $B \in \RR^{m\times n}$, $\lambda\in \RR$, $s\in \mathbb{Z}$\\
	{\bf output}: Sampling matrix $S \in \RR^{s \times D}$ 
	\begin{algorithmic}[1]
		
		\STATE{Generate $H\in\RR^{n\times d'}$ with i.i.d. normal entries, where $d'= C_1 q \log_2n$}
		\STATE{$M \gets \text{diag}\big( \{e^{-\|X_{\star,i}\|_2^2/2} \}_{i=1}^n \big) \cdot (B^\top B + \lambda I)^{-1/2} \cdot H$}\label{M-alg-gauss}
		\STATE{Let $Q^q\in\RR^{m'\times d^q}$ be an instance of the sketch from Lemma \ref{soda-result} with $\epsilon=\frac{1}{10q}$, $\delta=\frac{1}{\text{poly}(n)}$, $m'=C_2q^2 \log_2n$}
		\STATE{Compute $P_j=Q^q \left(X^{\otimes (q-j)} \otimes E_1^{\otimes j}\right)$ for all $j=0,1, \cdots q$, where $E_1=[{\bf e}_1, {\bf e}_1, \cdots {\bf e}_1] \in \RR^{d\times n}$}\label{Pj-alg-gauss}
		
		\STATE{$Z\gets \left( \frac{P_q}{\sqrt{0!}} \oplus \frac{P_{q-1}}{\sqrt{1!}} \oplus \frac{P_{q-2}}{\sqrt{2!}} \oplus \cdots \frac{P_0}{\sqrt{q!}} \right) \cdot M$}\label{Z-gauss}
		\STATE{$p_i\gets {\|Z_{\star,i}\|_2^2}/{\|Z\|_F^2}$ for every $i\in[d']$}\label{J-dist-alg-gauss}
		
		\STATE{Generate i.i.d. samples $j_1,j_2,\cdots j_s$ from  dist.  $\{p_i\}_{i=1}^{d'}$}
		
		\STATE{$h:[d]\rightarrow[s']$: fully independent hash with $s'=\lceil {q}^{\frac{3}{2}}s \rceil$}
		\STATE{Let $h^{-1}(r)=\{ j\in[d]:h(j)=r \}$ for every $r\in[s']$}
		
		\STATE{For all $r\in[s']$, generate $G_r \in \RR^{n'\times d_r}$ with i.i.d. normal entries, where $d_r=|h^{-1}(r)|$, $n'=C_3q^2\log_2n$}
		\STATE{$W_r \gets G_r \cdot X_{h^{-1}(r),\star}$ for every $r\in[s']$}\label{W-gauss}
		
		\FOR{$l=1$ to $s$}
		\STATE{$y_a \gets \frac{\| P_{q-a} \cdot M_{\star,j_l} \|_F^2/a!}{\sum_{b=0}^q \| P_{q-b} \cdot M_{\star,j_l} \|_F^2/b!}$ for every $a=0,1 \dots q$}\label{y-gauss}
		\STATE{Generate a sample $w$ from distribution $\{ y_a \}_{a=0}^q$}\label{w-gauss}

		\STATE{$D^{(0)} \gets \text{diag}(M_{\star,j_l})$}
		\FOR{$a=1$ to $w$}
		\STATE{$p^{a}_r\gets \frac{\left\|W_r \cdot D^{(a-1)} \cdot P_{a+q-w}^\top \right\|_F^2}{\sum_{t=1}^{s'} \left\|W_t \cdot D^{(a-1)} \cdot P_{a+q-w}^\top \right\|_F^2}$ for all $r \in [s']$}\label{dist-par-gauss}
		\STATE{Generate a sample $t$ from  distribution $\{p^a_r\}_{r=1}^{s'}$}
		\STATE{$q^a_i \gets \frac{\left\|X_{i,\star}  D^{(a-1)}  P_{a+q-w}^\top \right\|_2^2}{\left\|X_{h^{-1}(t),\star}  D^{(a-1)}  P_{a+q-w}^\top \right\|_F^2}$ for $i\in h^{-1}(t)$}\label{dist-qai-gauss}
		
		\STATE{Generate a sample $i_a$ from dist. $\{q^a_i \}_{i\in h^{-1}(t)}$}\label{isample-alg-gauss}
		
		\STATE{$D^{(a)}\gets D^{(a-1)}\cdot \text{diag}(X_{i_a,\star})$}
		\ENDFOR
		\STATE{$\beta \gets 0$}
		\FOR{$j=1$ to $d'$}\label{beta-forloop-gauss}
		\STATE{$y_w^* \gets \frac{\| P_{q-w} \cdot M_{\star,j} \|_F^2/w!}{\sum_{b=0}^q \| P_{q-b} \cdot M_{\star,j} \|_F^2/b!}$}\label{ystar-gauss}
		\STATE{$L^{(0)} \gets \text{diag}(M_{\star,j})$}
		\FOR{$b=1$ to $w$}
		\STATE{$p^*_b \gets \frac{\left\|W_{h(i_b)} \cdot L^{(b-1)} \cdot P_{b+q-w}^\top \right\|_F^2}{\sum_{t=1}^{s'} \left\|W_t \cdot L^{(b-1)} \cdot P_{b+q-w}^\top \right\|_F^2}$}\label{p-star-gauss}
		\STATE{$q^*_b \gets \frac{\left\|X_{i_b,\star}  L^{(b-1)}  P_{b+q-w}^\top \right\|_2^2}{\left\|X_{h^{-1}(h(i_b)),\star}  L^{(b-1)}  P_{b+q-w}^\top \right\|_F^2}$}\label{q-star-gauss}
		\STATE{$L^{(b)}\gets L^{(b-1)}\cdot \text{diag}(X_{i_b,\star})$}
		
		\ENDFOR
		\STATE{$\beta \gets \beta + sp_j y^*_w \cdot \prod_{b=1}^q (p^*_b q^*_b) $}
		\ENDFOR
		
		\STATE{$S_{l,\frac{d^w-1}{d-1}:\frac{d^{q+1}-d^{w+1}}{d-1}} \gets \frac{1}{\sqrt{\beta}} \left( {\bf e}_{i_1} \otimes {\bf e}_{i_2} \otimes  \cdots {\bf e}_{i_a} \right)^\top$ }\label{lthrow-gauss}
		\ENDFOR
		\STATE{\textbf{return} $S$}
	\end{algorithmic}
	\label{alg:rowsampler-gauss}
\end{algorithm}

\begin{lemma}\label{lem:rowsampler-gauss}
	For any matrix $X\in \RR^{d\times n}$, let $A=\phi_q(X)$ be a matrix with $n$ columns whose columns are obtained by applying the mapping $\phi_q$ as in Definition \ref{def:poly-feature-Gauss} to each column of $X$, i.e., $A_{\star,i}=\phi_q(X_{\star,i})$. For any matrix $B\in\RR^{m\times n}$, any $\lambda>0$, and any integers $s$ and $q$, Algorithm \ref{alg:rowsampler-gauss} outputs a rank-$s$ row norm sampler for matrix $A(B^\top B + \lambda I)^{-1/2}$ using $O\left( qm^2 n \log n + q^{15/2} s^2 n \log^3 n + q^{5/2} \log^3n \cdot \text{nnz}(X) \right)$ runtime.
\end{lemma}
\begin{proof}
	Let $A=\phi_q(X)$ be the matrix whose columns are $A_{\star,i} = \phi_q(X_{\star,i})$ for all $i$, where $\phi_q$ is the degree-$q$ polynomial lifting for Gaussian kernel as in Definition \ref{def:poly-feature-Gauss}.
	Algorithm \ref{alg:rowsampler-gauss} outputs a random sampling matrix $S\in\RR^{s\times D}$ with $D=\sum_{j=0}^{q} d^j$. First, we show that all rows of $S$ have independent and identical distributions. The reason is because the algorithm generates i.i.d. samples $j_1,j_2, \cdots j_s$ in line~7 and then for every $l\in[s]$, the $l^{th}$ row of the matrix $S$ is constructed by sampling $w$ in line~14 and then $i_1,i_2,\cdots i_w$ in line~20 from distributions that are solely determined by $j_l$ only, and is independent of the values of $j_{l'}$ for $l'\neq l$.
	Let $\widetilde{A} = A (B^\top B + \lambda I)^{-1/2}$. Now let us partition the matrix $\widetilde{A}$ as,
	\[\widetilde{A}=\begin{bmatrix}
	\widetilde{A}_0\\
	\widetilde{A}_1\\
	\vdots\\
	\widetilde{A}_q
	\end{bmatrix},\]
	where $\widetilde{A}_j$ is a $d^j\times n$ matrix for every $j=0,1,\cdots q$. Considering the action of the sampling matrix on matrix $\widetilde{A}$ will ease notation, so we consider the matrix $S\widetilde{A}$. Since every row of $S\widetilde{A}$ is identically distributed, let us consider the distribution of the $l^{th}$ row of $S\widetilde{A}$ for an arbitrary $l\in[s]$.
	
	Let $J$ be a random variable that takes values in $\{1,2, \cdots d'\}$ with probability distribution $\{p_i\}_{i=1}^{d'}$, which is defined in line~6 of Algorithm \ref{alg:rowsampler-gauss}. A random index $j_l$ generated in line~7 of the algorithm is a copy of the random variable $J$.
	For any $j\in[d']$, let $T^j$ be a random variable that takes values in $\{0,1, \cdots q\}$ with probability distribution
	\[ \Pr[T^j=a]= \frac{\| P_{q-a} \cdot M_{\star,j} \|_F^2/a!}{\sum_{b=0}^q \| P_{q-b} \cdot M_{\star,j} \|_F^2/b!}, \]
	where $P_{b}$ for $b=0,1, \cdots q$ are the matrices defined in line~4 and $M$ is the matrix defined in line~2 of the algorithm. The random sample $w$ generated in line~14 of the algorithm is a copy of the random variable $T^{j_l}$. 
	
	For any $j\in[d']$ let $I^{j,w}=(I^j_1,I^j_2, \cdots I^j_w)$ be a vector random variable that takes values in $[d]^w$ with the following conditional probability distribution for every $a=1,2, \cdots w$,
	\begin{align*}
	&\Pr[I^j_a =i | I^j_1=i_1, I^j_2=i_2, \cdots I^j_{a-1}=i_{a-1}]\\
	& = \frac{\|W_{h(i)} \cdot D^{a-1} P_{a+q-w}^\top \|_F^2}{\sum_{t=1}^{s'} \|W_{t}  D^{a-1} P_{a+q-w}^\top \|_F^2}  \frac{\| X_{i,\star} D^{a-1}P_{a+q-w}^\top \|_2^2}{\| X_{h^{-1}(h(i)),\star} D^{a-1}P_{a+q-w}^\top \|_F^2},
	\end{align*}
	where $W_r$ for $r\in[s']$ are the matrices defined in line~11 of the algorithm and $D^{a-1}$ is a diagonal matrix of size $n\times n$ whose diagonal entries are defined as,
	\[ D^{a-1}_{rr} = M_{r,j} \cdot \prod_{b=1}^{a-1} X_{i_b,r}, \]
	for every $a\in[q]$ and $r \in [n]$. For ease of notation we drop the superscript $w$ from $I^{j,w}$ and instead write $I^{j}$.
	It follows that the vector random variable $(i_1,i_2, \cdots i_q)$ obtained by putting together the random indices generated in line~20 of the algorithm, is a copy of the random variable $I^{j_l}$.
	
	Now we are ready to calculate the distribution of the $l^{th}$ row of $S\widetilde{A}$, which we denote by $[S\widetilde{A}]_{l,\star}$. 
	Let $\beta$ be the quantity that the for loop in lines~24-33 of the algorithm computes. If we let $w\in \{ 0,1, \cdots q \}$ be the random number generated in line~14 and if we let $i_1,i_2, \cdots i_q \in[d]$ be the indices sampled in line~20 of the algorithm, then we can compute the value of $\beta$ as follows,
	\[
	\beta = s\sum_{j=1}^{d'} \Pr[J=j] \Pr[T^j=w] \cdot\prod_{b=1}^q p^*_b q^*_b,
	\]
	where the quantities $p^*_b = \frac{\|W_{h(i_b)} \cdot D^{b-1} P_{b+q-w}^\top \|_F^2}{\sum_{t=1}^{s'} \|W_{t}  D^{b-1} P_{b+q-w}^\top \|_F^2}$ and $q^*_b = \frac{\| X_{i_b,\star} D^{b-1}P_{b+q-w}^\top \|_2^2}{\| X_{h^{-1}(h(i_b)),\star} D^{b-1}P_{b+q-w}^\top \|_F^2}$ are computed in lines~28 and 29 of the algorithm.
	Hence, the distribution of $[S\widetilde{A}]_{l,\star}$ is the following,
	\begin{align}
	&\Pr\left[ [S\widetilde{A}]_{l,\star} = \beta^{-1/2} \cdot [\widetilde{A}_w]_{(i_1,i_2, \cdots i_w),\star} \right]\nonumber\\
	&= \sum_{j=1}^{d'} \Pr\left[ \left. [S\widetilde{A}]_{l,\star} = \frac{[\widetilde{A}_w]_{(i_1, \cdots i_w),\star}}{\sqrt{\beta}} \right| w, j \right] \Pr[T^j=w]  p_j \nonumber\\
	&= \sum_{j=1}^{d'} \Pr\left[ I^j=(i_1 , i_2, \cdots i_w) \right] \Pr[T^j=w] \cdot p_j\nonumber\\
	&= \sum_{j=1}^{d'} \Pr[T^j=w]p_j \prod_{a=1}^{w} \Pr \left[I^j_a =i_a | i_1 i_2 \cdots i_{a-1} \right]  \label{dist-sampling-matrix-gauss}
	\end{align}
	where $\Pr \left[I^j_a =i_a | I^j_1=i_1, \cdots I^j_{a-1}=i_{a-1} \right] = p^*_a q^*_a$.  Therefore, $\Pr\left[ [S\widetilde{A}]_{l,\star}= \beta^{-1/2} \cdot [A_w]_{(i_1,i_2, \cdots i_w),\star} \right] = \frac{\beta}{s}$. 
	
	Now note that for any $r\in[s']$, the matrix $W_r$ is defined as $W_r=G_r \cdot X_{h^{-1}(r),\star}$ where $G_r$ is a matrix with i.i.d. Gaussian entries with $n'=C_3q^2\log_2n$ rows. Therefore $G_r$ is a Johnson-Lindenstrauss transform and hence for every $a\in[w]$ and every $r\in[s']$, the following holds with high probability,
	\begin{align}\label{jl-gauss}
	&\|W_{r} D^{a-1} P_{a+q-w}^\top \|_F^2\\ 
	&\qquad \in \left(1 \pm {0.1}/{q}\right) n' \cdot \| X_{h^{-1}(r),\star} D^{a-1} P_{a+q-w}^\top \|_F^2.\nonumber
	\end{align}
	By a union bound over $qs'd'$ events, \eqref{jl-gauss} holds simultaneously for all $a\in[w]$, all $j\in[d']$, and all $r\in[s']$ with high probability. We condition on \eqref{jl-gauss} holding in what follows. Therefore, we can bound the conditional probability of $I^j_a $ as follows,
	\begin{align*}
	&\Pr[I^j_a =i | I^j_1=i_1, I^j_2=i_2, \cdots I^j_{a-1}=i_{a-1}]\\
	& = \frac{\|W_{h(i)} \cdot D^{a-1} P_{a+q-w}^\top \|_F^2}{\sum_{t=1}^{s'} \|W_{t}  D^{a-1} P_{a+q-w}^\top \|_F^2}  \frac{\| X_{i,\star} D^{a-1}P_{a+q-w}^\top \|_2^2}{\| X_{h^{-1}(h(i)),\star} D^{a-1}P_{a+q-w}^\top \|_F^2}\\
	&\ge \frac{(1-\frac{1}{10q}) \| X_{i,\star} D^{a-1}P_{a+q-w}^\top \|_2^2}{(1+\frac{1}{10q}) \|X D^{a-1} P_{a+q-w}^\top \|_F^2}\\
	&\ge \left( 1-\frac{1}{5q} \right) \frac{ \| X_{i,\star} D^{a-1}P_{a+q-w}^\top \|_2^2}{ \|X D^{a-1} P_{a+q-w}^\top \|_F^2}.
	\end{align*}
	
	Now we invoke Lemma \ref{soda-result}. For every $b\in[q]$, $P_b$ is defined as $P_b=Q^q\left( X^{\otimes(q-b)} \otimes E_1^{\otimes b} \right)$, where $Q^q$ is the sketch from Lemma \ref{soda-result} with $\epsilon=\frac{1}{10q}$. We can write for every $i\in[d]$,
	\begin{align*}
	&P_{a+q-w} D^{a-1} X_{i,\star}^\top\\ 
	&\qquad= Q^q\left( X^{\otimes(w-a)} \otimes E_1^{\otimes a+q-w} \right) D^{a-1} X_{i,\star}^\top\\
	&\qquad= Q^q\left( \left(X^{\otimes(w-a)} D^{a-1} X_{i,\star}^\top\right) \otimes {\bf e}_1^{\otimes a+q-w} \right).
	\end{align*}
	Hence, if we invoke Lemma \ref{soda-result} we get that for every $a\in[w]$ and every $i\in[d]$, the following holds with high probability
	\begin{equation}\label{tensorsketch-gauss}
	\left\| X_{i,\star} D^{a-1}P_{a+q-w}^\top \right\|_2^2 \in \frac{\left\| X^{\otimes(w-a)} D^{a-1} X_{i,\star}^\top \right\|_2^2}{1\pm 0.1/q}.
	\end{equation}
	Moreover, 
	\[ P_{q-w} \cdot M_{\star,j}  = Q^q\left( \left(X^{\otimes w} M_{\star,j}\right) \otimes {\bf e}_1^{\otimes q-w} \right) , \]
	and hence, by Lemma \ref{soda-result}, for every $j \in [d']$ and every $w\in \{ 0,1, \cdots q \}$, the following holds with high probability,
	\begin{equation}\label{T-tensor-guarantee}
	\| P_{q-w} \cdot M_{\star,j} \|_F^2 \in \left(1\pm \frac{1}{10q} \right) {\| X^{\otimes w} \cdot M_{\star,j} \|_F^2}.
	\end{equation}
	By union bounding over $(q+1)d'(d+1)$ events we have that with high probability, both \eqref{tensorsketch-gauss} and \eqref{T-tensor-guarantee} hold simultaneously for all $w\in \{ 0,1, \cdots q \}$, all $a\in[w]$, all $j\in[d']$, and all $i\in[d]$. Therefore, conditioning on \eqref{tensorsketch-gauss} and \eqref{T-tensor-guarantee} holding, we have the following two bounds for the conditional probability of $I^j_a $ as well as the conditional probability $T^j$,
	\begin{align*}
	&\Pr[I^j_a =i | I^j_1=i_1, I^j_2=i_2, \cdots I^j_{a-1}=i_{a-1}]\\
	&\qquad\ge \left(1-\frac{2}{5q}\right)  \frac{ \| X^{\otimes(w-a)} D^{a-1} X_{i,\star}^\top \|_2^2}{ \|X^{\otimes(w-a)} D^{a-1} X^\top \|_F^2}.
	\end{align*}
	and,
	\begin{align}
	\Pr[T^j=w] &\in  \frac{\left( 1 \pm {1}/{5q}\right)\| X^{\otimes w} \cdot M_{\star,j} \|_F^2/w!}{\sum_{b=0}^q \| X^{\otimes b} \cdot M_{\star,j} \|_F^2/b!}\nonumber\\
	&= \left( 1 \pm {1}/{5q}\right) \frac{\left\| [\widetilde{A}_w H]_{\star,j} \right\|_2^2}{ \| [\widetilde{A} H]_{\star,j} \|_2^2} \label{prob-T-bound}
	\end{align}
	
	Also we use the following fact that follows from the definition of tensor products and the definition of matrix $D^a$,
	\begin{align*}
	&\left\| X^{\otimes(w-a-1)} D^{a} X^\top \right\|_F^2\\ &\qquad= \left\| X^{\otimes(w-a-1)} D^{a-1}\cdot \text{diag}(X_{i_a,\star}) X^\top \right\|_F^2\\
	&\qquad = \left\| X^{\otimes(w-a)} D^{a-1} X_{i_a,\star}^\top \right\|_2^2
	\end{align*}
	Now we compute the following product of the conditional probabilities
	\begin{align}
	&\prod_{a=1}^{w} \Pr \left[I^j_a =i_a | I^j_1=i_1, \cdots I^j_{a-1}=i_{a-1} \right]\nonumber\\
	&\qquad\ge \prod_{a=1}^{w} \left(1-\frac{2}{5q}\right)  \frac{ \| X^{\otimes(w-a)} D^{a-1} X_{i_a,\star}^\top \|_2^2}{ \|X^{\otimes(w-a)} D^{a-1} X^\top \|_F^2}\nonumber\\
	&\qquad\ge \frac{1}{2} \cdot \prod_{a=1}^{w} \frac{ \| X^{\otimes(w-a)} D^{a-1} X_{i_a,\star}^\top \|_2^2}{ \|X^{\otimes(w-a)} D^{a-1} X^\top \|_F^2}\nonumber\\
	&\qquad = \frac{1}{2} \cdot \frac{ \| X^{\otimes(0)} D^{q-1} X_{i_q,\star}^\top \|_2^2}{ \|X^{\otimes(w-1)} D^{0} X^\top \|_F^2}\nonumber\\
	&\qquad = \frac{1}{2} \cdot \frac{ \left| \langle [X^{\otimes w}]_{(i_1,i_2,\cdots i_w),\star}, M_{\star,j} \rangle \right|^2}{ \|X^{\otimes w} \cdot M_{\star,j} \|_2^2}\nonumber \\
	&\qquad = \frac{1}{2} \cdot \frac{ \left| [\widetilde{A}_w H]_{(i_1,i_2,\cdots i_w),j} \right|^2}{ \|[\widetilde{A}_w H]_{\star,j} \|_2^2} \label{conditional-prob-bound-gauss}
	\end{align}
	By plugging \eqref{prob-T-bound} and \eqref{conditional-prob-bound-gauss} back in 
	\eqref{dist-sampling-matrix-gauss} we get that,
	\begin{align}
	&\Pr\left[  [S\widetilde{A}]_{l,\star} = \beta^{-1/2} \cdot [\widetilde{A}_w]_{(i_1,i_2, \cdots i_w),\star} \right]\nonumber\\
	&\qquad\ge \sum_{j=1}^{d'} \frac{1-\frac{1}{5q}}{2}  \cdot \frac{\| [\widetilde{A}_w H ]_{\star,j} \|_2^2}{ \left\| [\widetilde{A} H]_{\star,j} \right\|_2^2} \frac{ \left| [\widetilde{A}_w H]_{(i_1,i_2,\cdots i_w),j} \right|^2}{ \left\|[\widetilde{A}_w H]_{\star,j} \right\|_2^2} p_j\nonumber\\
	&\qquad= \sum_{j=1}^{d'} \frac{1-\frac{1}{5q}}{2} \cdot \frac{ \left| [\widetilde{A}_w H]_{(i_1,i_2,\cdots i_w),j} \right|^2}{\left\| [\widetilde{A} H]_{\star,j} \right\|_2^2} p_j\label{dist-}
	\end{align}
	Now we bound $p_j$, which is defined in line~6 of the algorithm as follows,
	\begin{align*}
	p_j &= \frac{\|Z_{\star,j}\|_2^2}{\|Z\|_F^2}\\
	&= \frac{\sum_{b=0}^q \|P_{q-b} \cdot M_{\star,j}\|_2^2/b! }{\sum_{b=0}^q \|P_{q-b} \cdot M\|_F^2/b!}\\
	&\ge (1-1/5q) \frac{\sum_{b=0}^q \|X^{\otimes b} \cdot M_{\star,j}\|_2^2/b! }{\sum_{b=0}^q \|X^{\otimes b} \cdot M\|_F^2/b!}\\
	&= (1-1/5q) \frac{\left\| [\widetilde{A} H]_{\star,j} \right\|_2^2 }{\left\|\widetilde{A}H\right\|_F^2},
	\end{align*}
	where the inequality above follows from \eqref{T-tensor-guarantee}. Plugging the above into \eqref{dist-}, we get that,
	\begin{align*}
	&\Pr\left[  [S\widetilde{A}]_{l,\star} = \beta^{-1/2} \cdot [\widetilde{A}_w]_{(i_1,i_2, \cdots i_w),\star} \right]\\
	&\ge \sum_{j=1}^{d'} \frac{1-\frac{2}{5q}}{2} \cdot \frac{ \left| [\widetilde{A}_w H]_{(i_1,i_2,\cdots i_w),j} \right|^2}{\left\| [\widetilde{A} H]_{\star,j} \right\|_2^2} \frac{\left\| [\widetilde{A} H]_{\star,j}\right\|_2^2 }{\left\|\widetilde{A}H\right\|_F^2}\\
	&\ge \frac{1}{3} \cdot \sum_{j=1}^{d'} \frac{ \left| [\widetilde{A}_w H]_{(i_1,i_2,\cdots i_w),j} \right|^2}{\left\|\widetilde{A}H\right\|_F^2}\\
	&= \frac{1}{3} \cdot \frac{ \left\| [\widetilde{A}_w H]_{(i_1,i_2,\cdots i_w),\star} \right\|_2^2}{\left\|\widetilde{A}H\right\|_F^2}
	\end{align*}
	
	Now note that $H$ is a matrix with i.i.d. Gaussian entries with $d' = C_1q\log_2n$ columns, and therefore $H$ is a Johnson-Lindenstrauss transform, and hence for every $ w \in \{0,1, \cdots q \}$ and every $(i_1,i_2,\cdots i_w) \in[d]^w$, with probability $1-\frac{1}{\text{poly}(n^{q+1})}$,
	\[\left\| [\widetilde{A}_w H ]_{(i_1,i_2,\cdots i_w),\star} \right\|_2^2 \in d' \left(1\pm 0.1\right) \left\| [\widetilde{A}_w ]_{(i_1,i_2,\cdots i_w),\star} \right\|_2^2.
	\]
	Therefore, by union bounding over $d^{q+1}$ events, the above holds simultaneously for all $w$ and all $(i_1,i_2,\cdots i_w) \in[d]^w$ with high probability. Therefore,
	\begin{align*}
	&\Pr\left[ [S\widetilde{A}]_{l,\star} = \beta^{-1/2} \cdot [\widetilde{A}_w]_{(i_1,i_2, \cdots i_w),\star} \right]\\
	&\qquad\ge \frac{1}{4} \cdot \frac{ \left\| [\widetilde{A}_w]_{(i_1,i_2,\cdots i_w),\star} \right\|_2^2}{\|A\|_F^2}.
	\end{align*}
	Because $\frac{\beta}{s}$ is the probability of sampling row $(i_1,i_2,\cdots i_w)$ of the $w^{th}$ block of the matrix $\widetilde{A}$, the above inequality proves that with high probability, $S$ is a rank-$s$ row norm sampler for $\widetilde{A}$ as in Definition \ref{def:row-samp}. 
	
	\paragraph{Runtime:} The operations that this algorithm perform largely overlap with that of Algorithm \ref{alg:rotatedrowsampler-poly} with a few additional operations. One of the additional computations in this algorithm is the computation of the matrix $Z$ in line~5 of the algorithm, which takes $O(q^4 n \log^2 n)$ operations. Another additional computational part of the algorithm is the computation of $y_a$, for $a=0,1,\cdots q$, in line~13 of the algorithm, that can be computed in time $O\left( q^3 n \log n \right)$ time for a fixed $l \in[s]$. Therefore, the total time to compute this distribution for all $l$ is $O\left( q^3 s n \log n \right)$. 
	Finally the last additional computation is the computation of the quantity $y^*_w$ in line~25 of the algorithm which takes time $O(q^3 n \log n)$ for a fixed $l \in [s]$ and a fixed $j \in [d']$. Hence the total time of this operation for all $l$ and $j$ is $O\left( q^4 s n \log^2 n \right)$. The total runtime of Algorithm \ref{alg:rowsampler-gauss} is the sum of these terms and the runtime of Algorithm \ref{alg:rotatedrowsampler-poly}, which results in $O\left( qm^2 n \log n + q^{15/2} s^2 n \log^3 n + q^{5/2} \log^3n \cdot \text{nnz}(X) \right)$ runtime.
\end{proof}

\section{Proof of Theorem \ref{thm:gauss}}
Let $q = C (r+\log_2n)$ for a large enough constant $C$. Let $\phi_q$ be the degree-$q$ polynomial lifting for the Gaussian kernel as in Definition \ref{def:poly-feature-Gauss}. Let $\Phi$ be the matrix with $n$ columns whose columns are obtained by applying the lifting $\phi_q$ on the data points, i.e., $\Phi_{\star,i} = \phi_q(x_i)$ for all $i\in[n]$. First of all, note that by Claim \ref{claim:poly-lift-gauss}, since we assumed $\epsilon , \lambda \ge \frac{1}{\text{poly}(n)}$,
\[\|\Phi^\top \Phi - K\|_{op} \le \frac{\epsilon}{2} \lambda.\]
The algorithm finds a spectrally close surrogate for the Gaussian kernel matrix $K$ by invoking the recursive leverage score sampling procedure of Algorithm \ref{alg:outerloop} with inputs $\Phi$, $\lambda$, $\epsilon/2$, and $\mu=O(s_\lambda)$.
For every invocation of the primitive \textsc{RowNormSampler} by Algorithm \ref{alg:outerloop}, we run Algorithm \ref{alg:rowsampler-gauss}, which is especially designed to perform row norm sampling on the Gaussian kernel's polynomial lifting matrix $\Phi$. By Lemma \ref{lem:rowsampler-gauss}, for any $\lambda'>0$, any integers $q,s'$ and any matrices
$X, B$, with high probability, the procedure \textsc{RowNormSampler}$(X,q,B,\lambda',s')$ of Algorithm \ref{alg:rowsampler-gauss} outputs a rank-$s'$ row norm sampler for matrix $\Phi(B^\top B +\lambda' I)^{-1/2}$. Therefore, because the total number of times Algorithm \ref{alg:rowsampler-gauss} is invoked by the recursive leverage score sampling procedure is bounded by $O\left( \log \frac{\|\phi\|_F^2}{\epsilon\lambda} \right) = O\left( \log \frac{{\bf tr}(K)}{\epsilon\lambda} \right) = O(\log n)$, by a union bound, with high probability the preconditions of Lemma \ref{lem:rowsampler-gauss} hold and hence we can invoke this lemma to prove that the sampler $\Pi$ returned by Algorithm \ref{alg:outerloop} satisfies the following with high probability:
\[ \frac{\Phi^\top \Phi + \lambda I}{1+\epsilon/2} \preceq \Phi^\top \Pi^\top \Pi \Phi + \lambda I \preceq \frac{\Phi^\top \Phi + \lambda I}{1-\epsilon/2}.\] 
Therefore, since $\|\Phi^\top \Phi - K\|_{op} \le \frac{\epsilon}{2} \lambda$, this implies,
\[ \frac{K + \lambda I}{1+\epsilon} \preceq \Phi^\top \Pi^\top \Pi \Phi + \lambda I \preceq \frac{K + \lambda I}{1-\epsilon}.\] 

Therefore, if we let $Z=\Pi \Phi$, the theorem follows because $\Pi$ has $s=O\left( \frac{s_\lambda}{\epsilon^2} \log n\right)$ rows. Also because Algorithm~\ref{alg:outerloop} calls the primitive \textsc{RowSampler}$(X,q,B,\lambda',s')$ of Algorithm~\ref{alg:rotatedrowsampler-poly}, $O\left( \log \frac{\|\phi\|_F^2}{\epsilon\lambda} \right) = O\left( \log \frac{{\bf tr}(K)}{\epsilon\lambda} \right) = O(\log n)$ times with inputs $s' = O(\frac{s_\lambda}{\epsilon^2} \log n)$ and matrix $B$, which has $O(s_\lambda \log n)$ rows, each call, by Lemma \ref{lem:rowsampler-gauss}, takes $O\left( {\text{poly}(\epsilon^{-1},q,\log n) \cdot s_\lambda^2n} + q^{5/2} \log^3n \cdot \text{nnz}(X) \right)$ operations. Hence, since $q\approx r+\log n$, the total runtime of the algorithm is $O\left( {\text{poly}(\epsilon^{-1},r,\log n) \cdot s_\lambda^2n} + r^{5/2} \log^4n \cdot \text{nnz}(X) \right)$.

\section{Kernel Ridge Regression}

One of the most elementary and yet powerful kernel methods is {\em Kernel Ridge Regression (KRR)}.
Given training data $(x_1, y_1),\dots,(x_n,y_n)\in \RR^d\times\RR$, a kernel function $k:\RR^d\times\RR^d\to\RR$, and a regularization
parameter $\lambda > 0$, the KRR estimator for a given input $x$ is:
\begin{equation*}
\bar{f}(x) \equiv \sum^n_{j=1} k(x_j, x) \alpha_j
\end{equation*}
where ${\bm \alpha} = ( \alpha_1  \cdots  \alpha_n)^\top$ is the solution to the equation,
\begin{equation}
\label{eq:linear}
(K + \lambda I){\bm \alpha} = {\bf y}.
\end{equation}
In \eqref{eq:linear}, $K\in\RR^{n\times n}$ is the kernel matrix defined by $K_{ij}\equiv k(x_i, x_j)$
and ${\bf y} \equiv [y_1 \cdots y_n]^\top$ is the vector of responses.
The KRR estimator can be derived by minimizing
a regularized squared loss objective function over a hypothesis space defined by the reproducing kernel Hilbert space associated with
$k(\cdot)$. However, the details are not important here.

Suppose that $\phi$ is the lifting corresponding to the kernel function, i.e., $k(x,z) = \langle \phi(x), \phi(z) \rangle$. Let $\Phi$ be the matrix with $n$ columns which is obtained by applying the lifting $\phi$ on the dataset, i.e.,  $\Phi_{\star,i} = \phi(x_i)$. Then, Theorems \ref{thm:poly} and \ref{thm:gauss} approximate the kernel matrix $K$ by finding a sampling matrix $\Pi$ such that $\Phi^\top \Pi^\top \Pi \Phi \approx \Phi^\top \Phi = K$. This corresponds to approximating the kernel function $k(\cdot)$ by $\tilde{k}(x,z) = \langle \Pi\phi(x), \Pi\phi(z) \rangle$. Therefore, the approximate KRR estimator for a given input $x$ is,
\[ \tilde{f}(x) \equiv \sum^n_{j=1} \tilde{k}(x_j, x) \tilde{\alpha}_j = \langle {\bf w}, \Pi \phi(x) \rangle, \] 
where the vector ${\bf w}$ is obtained by solving the equation,
\[ (\Pi \Phi \Phi^\top \Pi^\top + \lambda I){\bf w} = \Pi \Phi{\bf y}. \]
The above equation can be solved much faster than \eqref{eq:linear} since the sampling matrix $\Pi$ has a small number $s \approx \frac{s_\lambda}{\epsilon^2} \log n$ of rows.
\subsection{Risk Bounds}
\label{sec:risk-bounds}

One way to analyze our approximate KRR estimator is via risk bounds. Several recent papers on approximate KRR use
such analysis~\cite{bach2013sharp, alaoui2015fast,musco2017recursive,avron2017random}.
In particular, these papers
consider the fixed design setting and upper bound the expected in-sample predication error of
the KRR estimator $\bar{f}$,
considering it as an empirical estimate of the statistical risk. More precisely, the underlying
assumption is that $y_i$ satisfies
\begin{equation}
\label{eq:stat-model}
y_i = f^*(x_i) + \nu_i
\end{equation}
for some $f^\star:\RR^d \to \RR$. The $\{\nu_i\}$'s are i.i.d noise terms, distributed as normal variables
with variance $\sigma^2_\nu$. The empirical risk of an estimator $f$, which measures the quality of the estimator, is defined as
\[\mathcal{R}(f) \equiv \EE_{\{\nu_i\}} \left[{\frac{1}{n}\sum^n_{j=1}\left| f(x_i) - f^*(x_i) \right|^2}\right].\]

Let ${\bf f} \in \RR^n$ be the vector whose $j^{th}$ entry is $f^*(x_j)$.
It is straightforward to show that for the KRR estimator $\bar{f}$ we have~\cite{bach2013sharp, alaoui2015fast,avron2017random}:
\begin{align*}
\mathcal{R}(\bar{f}) & = n^{-1}\lambda^2 {\bf f}^\top (K + \lambda I)^{-2} {\bf f} \\
& \quad + n^{-1}\sigma_\nu^2 \cdot {\bf tr}\left({K^2(K + \lambda I)^{-2}}\right).
\end{align*}
Since $\lambda^2 {\bf f}^\top (K + \lambda I)^{-2} {\bf f} \leq  \lambda {\bf f}^\top (K + \lambda I)^{-1} {\bf f}$
and ${\bf tr}\left({K^2(K + \lambda I)^{-2}}\right) \leq {\bf tr}\left({K(K + \lambda I)^{-1}}\right) = s_\lambda$, where $s_\lambda$ is the statistical dimension of the kernel matrix $K$.
We define,
$$\widehat{\mathcal{R}}_{K}({\bf f}) \equiv  n^{-1}\lambda {\bf f}^\top (K + \lambda I)^{-1} {\bf f} + n^{-1}\sigma_\nu^2 \cdot s_\lambda$$
and note that $\mathcal{R}(\bar{f}) \leq \widehat{\mathcal{R}}_{K}({\bf f})$. The first term in the above expressions
for $\mathcal{R}(\bar{f})$ and $\widehat{\mathcal{R}}_{K}({\bf f})$ is frequently referred to as \emph{bias},
while the second term is the \emph{variance}.

\begin{lemma}
	\label{lem:risk-bound}
	Suppose that~\eqref{eq:stat-model} holds, and let ${\bf f} \in \RR^n$ be the vector whose $j^{th}$ entry is $f^*(x_j)$.
	Let $\bar{f}$ be the KRR estimator, and let $\tilde{f}$ be the KRR estimator
	obtained using some other kernel $\tilde{k}(\cdot,\cdot)$, whose kernel matrix is $\widetilde{K}$.
	Suppose that $\widetilde{K}$ is an $(\epsilon,\lambda)$-spectral approximation to $K$ as in \eqref{spectral-bound} for some $\epsilon<1$,
	and that ${\|K\|_{op}} \geq 1$. The following bound holds:
	\begin{align}\label{estimateUpperBound}
	\mathcal{R}(\tilde{f}) \leq (1-\epsilon)^{-1} \widehat{\mathcal{R}}_{K}({\bf f}) + \frac{\epsilon}{1+\epsilon}\cdot\frac{\text{rank}{(\widetilde{K})}}{n}\cdot\sigma^2_\nu
	\end{align}
\end{lemma}
\begin{proof}
	For the bias term we have:
	\begin{align}\label{easyBiasBound}
	{\bf f}^\top (\widetilde{K} + \lambda I)^{-1} {\bf f} \leq (1-\epsilon)^{-1}{\bf f}^\top (K + \lambda I)^{-1} {\bf f}.
	\end{align}
	
	We now consider the variance term. Denote $s = \text{rank}({\widetilde{K}})$, and let
	$\lambda_1(A) \geq \lambda_2(A) \geq \dots \geq \lambda_n(A)$ denote the eigenvalues
	of a matrix $A$. We have:
	\begin{align*}
	s_\lambda(\widetilde{K})  &= {\bf tr}\left({(\widetilde{K} + \lambda I)^{-1}\widetilde{K}}\right)\\ 
	&= \sum^s_{i=1}\frac{\lambda_i(\widetilde{K})}{\lambda_i(\widetilde{K}) + \lambda} \\
	&= s -  \sum^s_{i=1}\frac{\lambda}{\lambda_i(\widetilde{K}) + \lambda} \\
	&\leq s - (1+\epsilon)^{-1} \sum^s_{i=1}\frac{\lambda}{\lambda_i(K) + \lambda} \\
	&= s -  \sum^s_{i=1}\frac{\lambda}{\lambda_i(K) + \lambda} + \frac{\epsilon}{1+\epsilon}\sum^s_{i=1}\frac{\lambda}{\lambda_i(K) + \lambda} \\
	&\leq n -  \sum^n_{i=1}\frac{\lambda}{\lambda_i(K) + \lambda} + \frac{\epsilon\cdot s}{1+\epsilon} \\
	&= s_\lambda(K) + \frac{\epsilon\cdot s}{1+\epsilon} \\
	&\leq (1-\epsilon)^{-1}s_\lambda(K) + \frac{\epsilon\cdot s}{1+\epsilon}
	\end{align*}
	where we use the fact that ${A} \preceq {B}$ implies that $\lambda_i(A) \leq \lambda_i(B)$
	(this is a simple consequence of the Courant-Fischer minimax theorem).
	
	Combining the above variance bound with the bias bound in \eqref{easyBiasBound} yields:
	$$\widehat{\mathcal{R}}_{\tilde{K}}({\bf f}) \leq (1-\epsilon)^{-1} \widehat{\mathcal{R}}_{K}({\bf f}) + \frac{\epsilon}{(1+\epsilon)}\cdot\frac{\text{rank}({\widetilde{K}})}{n}\cdot\sigma^2_\nu$$
	and the bound $\mathcal{R}(\tilde{f}) \leq \widehat{\mathcal{R}}_{\widetilde{K}}({\bf f})$ completes the
	proof.
	
\end{proof}

In short, Lemma \ref{lem:risk-bound} bounds the risk of the approximate KRR estimator as a function of both the risk upper bound $\widehat{\mathcal{R}}_{K}({\bf f})$ in \eqref{estimateUpperBound} and an additive term which is small if the rank of $\text{rank}({\widetilde{K}})$ and/or $\epsilon$ is small. In particular, it is instructive to compare the additive term $\frac{\epsilon}{1+\epsilon} \cdot n^{-1}\sigma_\nu^2 \cdot \text{rank}({\widetilde{K}})$ to the variance term $n^{-1}\sigma_\nu^2 \cdot s_\lambda$. Since the approximation $\widetilde{K}$ is only useful computationally if $\text{rank}({\widetilde{K}}) \ll n$, we should expect the additive term in \eqref{estimateUpperBound} to also approach $0$ and generally be small when $n$ is large.
%%%%%%%%%%%%%%%%%%%%%%%%%%%%%%%%%%%%%%%%%%%%%%%%%%%%%%%%%%%%%%%%%%%%%%%%%%%%%%%
%%%%%%%%%%%%%%%%%%%%%%%%%%%%%%%%%%%%%%%%%%%%%%%%%%%%%%%%%%%%%%%%%%%%%%%%%%%%%%%

\end{document}